\documentclass[11pt]{article}
\pdfoutput=1
\usepackage[margin=1in]{geometry}
\usepackage[onehalfspacing]{setspace}
\usepackage{graphicx,amssymb}
\usepackage{amsmath,mathtools,amsthm}
\usepackage{float}
\usepackage{xfrac}
\usepackage{enumerate,multicol,multirow}
\usepackage{caption,subcaption}
\usepackage{hyperref}
\usepackage[capitalize]{cleveref}
\usepackage{algorithm}
\usepackage{algorithmic}
\usepackage{tikz}
\usetikzlibrary{decorations.pathreplacing}

\usepackage{thm-restate}
\usepackage{natbib}
\usepackage{xcolor}

\usepackage{color-edits}
\addauthor{yl}{red}    
\addauthor{gx}{blue}

\usepackage{todonotes}

\newtheorem{assumption}{Assumption}
\newtheorem{definition}{Definition}

\newtheorem{theorem}{Theorem}
\newtheorem{lemma}{Lemma}

\newtheorem{proposition}{Proposition}

\newtheorem{example}{Example}
\newtheorem{remark}{Remark}

\newcommand{\notshow}[1]{{}}

\newcommand{\Xcomment}[1]{{}}

\newcommand{\given}{\,\vert\,}

\newcommand{\prob}[2][]{\text{\bf Pr}\ifthenelse{\not\equal{}{#1}}{_{#1}}{}\!\left[{\def\givenn{\middle|}#2}\right]}
\newcommand{\expect}[2][]{\text{\bf E}\ifthenelse{\not\equal{}{#1}}{_{#1}}{}\!\left[{#2}\right]}

\newcommand{\tparen}{\big}
\newcommand{\tprob}[2][]{\text{\bf Pr}\ifthenelse{\not\equal{}{#1}}{_{#1}}{}\tparen[{\def\given{\tparen|}#2}\tparen]}
\newcommand{\texpect}[2][]{\text{\bf E}\ifthenelse{\not\equal{}{#1}}{_{#1}}{}\tparen[{\def\given{\tparen|}#2}\tparen]}

\newcommand{\sprob}[2][]{\text{\bf Pr}\ifthenelse{\not\equal{}{#1}}{_{#1}}{}[#2]}
\newcommand{\sexpect}[2][]{\text{\bf E}\ifthenelse{\not\equal{}{#1}}{_{#1}}{}[#2]}

\newcommand{\indicate}[1]{\mathbf{1}\left[#1\right]}

\newcommand{\regret}{{\rm Reg}}
\newcommand{\robust}{{\rm RMM}}
\newcommand{\oih}{{\rm OIH}}

\newcommand{\edp}{\hat{\mathbf{P}}}

\title{Robust Aggregation of Calibrated Forecasts}

\author{Xinxiang Guo \and Yingkai Li \and Yifen Mu}

\begin{document}
\maketitle

\begin{abstract}
    Decision-makers often rely on multiple probabilistic forecasts that are individually calibrated but need not be fully informative. We develop a framework for aggregating such forecasts when the decision-maker knows only that experts satisfy calibration. We show that the joint distribution of calibrated forecasts can contain decision-relevant information that is unavailable from any single expert, so the standard optimal-in-hindsight (OIH) benchmark may substantially understate attainable performance. To formalize this idea, we introduce a robust max-min benchmark: the best payoff a decision-maker can guarantee against all profile-wise conditional-mean mappings compatible with calibration. This benchmark is tractable, admits a linear-programming formulation, and dominates the OIH benchmark up to calibration error. It can nevertheless be strictly below the Bayesian benchmark, clarifying the value of knowing experts' information structures. Finally, we provide online algorithms that attain the robust benchmark under forecast-only feedback and stronger contextual benchmarks under state feedback.
\end{abstract}

\section{Introduction}\label{sec:introduction}

Probabilistic forecasts are increasingly used to guide decisions in settings such as weather prediction, medical diagnosis, financial investment, and policy evaluation. In these environments, a decision-maker often observes reports from multiple forecasters---experts, statistical models, or institutions---and must translate these reports into payoff-relevant actions. A natural minimal requirement on such forecasts is calibration: when a forecaster repeatedly reports probability $p$, the event should occur with frequency approximately $p$. Calibration is attractive because it is model-free and does not require the decision-maker to know the forecaster's information structure.

Calibration, however, is not the same as informativeness \citep[see, e.g.,][]{hu2024predict}. A forecaster who always reports the unconditional frequency of an event may be perfectly calibrated but useless for choosing state-contingent actions. This limitation is especially salient when forecasts are evaluated one at a time. In many applications, the decision-maker observes several calibrated forecasts simultaneously. The joint pattern of these forecasts may contain information that is absent from any individual report. For example, two forecasts that are each only partially informative may, through their correlation with each other and with the state, jointly reveal substantially sharper information about the relevant posterior belief. The central question of this paper is whether, and how, a decision-maker can exploit such information using calibration alone.


This question is also about the appropriate benchmark for online decision-making with forecasts. A standard benchmark in prediction with expert advice is the optimal-in-hindsight benchmark, which compares the decision-maker with the decision rule based on a single expert's reports---that is, the best mapping from one expert's report space to actions---selected in hindsight. This benchmark is natural when the learner's task is to select among experts. It is less compelling when the learner observes all forecasts and can aggregate them. If calibrated forecasts jointly encode decision-relevant information, then a decision rule that conditions on the full forecast profile may systematically outperform every rule that uses only one expert's report. In such cases, the optimal-in-hindsight benchmark understates what is achievable from the information available to the decision-maker.


A simple example illustrates the force of aggregation. Consider a weather-prediction problem over $T$ days. Expert~1 reports a probability of $1/3$ in the first $2T/3$ days and a probability of $2/3$ in the remaining $T/3$ days, while Expert~2 reports a probability of $2/3$ in the first $T/3$ days and a probability of $1/3$ in the remaining $2T/3$ days. See \Cref{fig:toy_example} for an illustration of this example. Suppose both experts are calibrated. Taken separately, each forecast partitions the data only coarsely. Taken together, however, the two calibration constraints identify a sharper posterior: the implied probability of rain is $2/3$ in the first third, $0$ in the middle third, and $2/3$ in the last third. 
To see why this matters for decisions, suppose the decision-maker must choose whether to prepare for rain, which is optimal exactly when the probability of rain exceeds $1/2$. Conditioning on the joint profile, the decision-maker prepares in the first and last thirds and refrains in the middle third. No rule based on a single expert's report can replicate this: each expert reports $1/3$ on a block of days that pools the middle third, where rain never occurs, with another third where rain occurs with probability $2/3$, so any single-expert rule must take the same action across that entire block.
The joint pattern of reports therefore reveals payoff-relevant information that is unavailable from either expert alone. This example shows why a benchmark based on the best individual expert can be too weak: the decision-maker's real informational object is the forecast profile, not a single forecast.

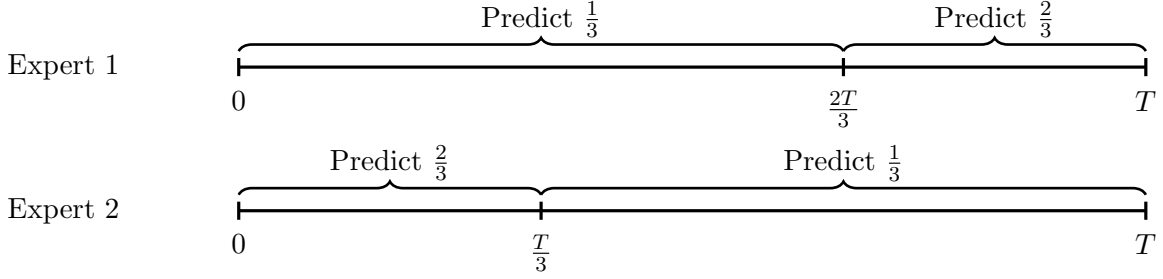
\begin{figure}[t]
\centering
\begin{tikzpicture}[x=12cm,y=1cm, >=latex]
 
\def\yone{1.9}
\def\ytwo{0.0}

\node[anchor=east] at (-0.12,\yone) {Expert 1};

\draw[line width=1.2pt] (0,\yone) -- (1,\yone);

\draw[line width=1.2pt] (0,\yone+0.12) -- (0,\yone-0.12);
\draw[line width=1.2pt] (2/3,\yone+0.12) -- (2/3,\yone-0.12);
\draw[line width=1.2pt] (1,\yone+0.12) -- (1,\yone-0.12);

\node[anchor=north] at (0,\yone-0.18) {$0$};
\node[anchor=north] at (2/3,\yone-0.18) {$\frac{2T}{3}$};
\node[anchor=north] at (1,\yone-0.18) {$T$};

\draw[decorate,decoration={brace,amplitude=5pt,raise=6pt},line width=1pt]
  (0,\yone) -- (2/3,\yone)
  node[midway, yshift=18pt] {Predict $\frac{1}{3}$};

\draw[decorate,decoration={brace,amplitude=5pt,raise=6pt},line width=1pt]
  (2/3,\yone) -- (1,\yone)
  node[midway, yshift=18pt] {Predict $\frac{2}{3}$};

\node[anchor=east] at (-0.12,\ytwo) {Expert 2};

\draw[line width=1.2pt] (0,\ytwo) -- (1,\ytwo);

\draw[line width=1.2pt] (0,\ytwo+0.12) -- (0,\ytwo-0.12);
\draw[line width=1.2pt] (1/3,\ytwo+0.12) -- (1/3,\ytwo-0.12);
\draw[line width=1.2pt] (1,\ytwo+0.12) -- (1,\ytwo-0.12);

\node[anchor=north] at (0,\ytwo-0.18) {$0$};
\node[anchor=north] at (1/3,\ytwo-0.18) {$\frac{T}{3}$};
\node[anchor=north] at (1,\ytwo-0.18) {$T$};

\draw[decorate,decoration={brace,amplitude=5pt,raise=6pt},line width=1pt]
  (0,\ytwo) -- (1/3,\ytwo)
  node[midway, yshift=18pt] {Predict $\frac{2}{3}$};

\draw[decorate,decoration={brace,amplitude=5pt,raise=6pt},line width=1pt]
  (1/3,\ytwo) -- (1,\ytwo)
  node[midway, yshift=18pt] {Predict $\frac{1}{3}$};

\end{tikzpicture}
\caption{An illustrative example.}
\label{fig:toy_example}
\end{figure}

The example is deliberately stark: the two calibration constraints uniquely determine the relevant posterior belief for each forecast profile. In general, calibration need not be so informative. The same joint distribution of calibrated forecasts may be consistent with many different mappings from forecast profiles to state probabilities. Thus, calibration can reveal useful restrictions on posterior beliefs without identifying them completely. This creates a basic ambiguity for decision-making: when several posterior mappings are consistent with all observed calibration constraints, which one should the decision-maker use?


We resolve this ambiguity by adopting a robust perspective. For each forecast profile, we consider the set of posterior probabilities that are consistent with the calibration of all experts. A decision rule maps forecast profiles into actions, and its value is evaluated against the worst posterior mapping in this feasible set. The robust max-min benchmark is the highest payoff that can be guaranteed using only the restrictions imposed by calibration. This benchmark captures the informational content of calibrated forecasts without requiring the decision-maker to know the experts' signal structures, priors, or conditional independence relationships.

This benchmark has two useful interpretations. First, it is an aggregation benchmark: it allows the decision-maker to condition on the full profile of reports and therefore to exploit information contained in the joint pattern of forecasts. Second, it is a robustness benchmark: it does not assume that calibration identifies the true posterior, but instead asks what payoff is guaranteed uniformly over all posterior beliefs compatible with calibration. In this sense, the benchmark lies between two familiar extremes. It is more demanding than the optimal-in-hindsight benchmark, which restricts attention to one expert at a time, but it is generally less demanding than the Bayesian benchmark, which assumes knowledge of the full data-generating process.


Our first set of results studies this benchmark in the offline environment, where the full sequence of forecast profiles is fixed. We first show that the robust benchmark is a tractable object. Although the benchmark is defined as a max-min problem over decision rules and calibrated-consistent posterior mappings, calibration imposes only linear restrictions. By dualizing the inner minimization problem, we obtain an equivalent linear program. Thus, the benchmark can be computed in time polynomial in the number of forecast profiles, the number of actions, and the number of distinct expert reports (\cref{prop_compute_RMM}). This tractability result is important because the benchmark is not merely conceptual: it can be evaluated and used as a practical performance criterion.

The offline analysis also clarifies the structure of optimal robust aggregation. A priori, one might think that the decision-maker should use correlated randomizations across different forecast profiles, because the calibration constraints couple the posterior probabilities assigned to different profiles. We show that allowing such correlation provides no gain: independent mixed strategies, chosen separately for each forecast profile, achieve the same robust value as arbitrary correlated strategies over mappings from forecast profiles to actions (\cref{thm_indepen_strate_optimal}). This result justifies the simpler formulation of a decision rule as a collection of profile-contingent mixed actions.

We then compare the robust benchmark with two natural alternatives. Relative to the optimal-in-hindsight benchmark, robust aggregation is weakly stronger when forecasts are exactly calibrated, and remains stronger up to an error term proportional to the calibration error (\cref{prop:robust_vs_oih}). This formalizes the intuition from the motivating example: a decision-maker who observes all forecasts should not be evaluated only against the best single expert. At the same time, the robust benchmark is generally weaker than the Bayesian benchmark, which assumes knowledge of the prior and the experts' information structures. We show in \cref{ex_cond_indep} where this inequality is strict. The robust benchmark therefore occupies an intermediate position: it extracts all information guaranteed by calibration, but it does not attribute to the decision-maker information about the data-generating process that calibration alone does not reveal.


We next study whether robust aggregation can be learned online. This question is nontrivial because, under forecast-only feedback, the decision-maker observes the sequence of forecast profiles but neither the realized states nor the realized payoffs. Standard online-learning algorithms rely on outcome or payoff feedback to update decisions. Here, by contrast, learning must proceed entirely from the empirical distribution of forecasts. The relevant object to learn is not which action performed well in the past, but which robust decision rule is warranted by the calibration restrictions induced by the distribution of forecast profiles.

Our main online result shows that this is possible. We consider a forecast-only feedback environment in which forecast profiles are drawn from a fixed distribution, while the realized state sequence is otherwise unrestricted subject to calibration (\cref{thm_online_regret}). We propose a plug-in robust aggregation algorithm: at each period, the decision-maker estimates the distribution of forecast profiles from past observations and solves a penalized version of the robust max-min problem for the estimated distribution. The penalty accommodates the fact that finite-sample empirical distributions may not exactly satisfy the calibration restrictions. We prove that this algorithm achieves sublinear robust regret. Hence, even without observing states or payoffs, the decision-maker can asymptotically attain the robust benchmark using only the observed sequence of calibrated forecasts. 
Finally, we show that sublinear regret with respect to the robust benchmark is not possible without the stochastic forecast profile assumption (see \cref{subsec:impossibility}). 

Note that the forecast-only model captures environments in which outcomes are delayed, unavailable, or not directly attributable to individual decisions. In other settings, the decision-maker observes the realized state after acting. This additional information permits a stronger benchmark. Under state feedback, the forecast profile can be treated as a context, and observing the state allows the decision-maker to evaluate the payoff of every action for that context. We therefore compare the learner with the best context-dependent decision rule in hindsight. A simple contextual Hedge algorithm, run separately for each forecast profile, achieves sublinear regret relative to this benchmark (\cref{thm_contextual_regret_full_feedback}). Thus, the appropriate benchmark depends sharply on the feedback structure: forecast-only feedback calls for robust aggregation based on calibration, whereas state feedback permits the stronger contextual optimal-in-hindsight benchmark.

Taken together, these results provide a framework for decision-making with calibrated forecasts that separates three sources of value: the value of aggregating multiple calibrated reports, the value of robustness when calibration does not identify a unique posterior, and the value of feedback in online learning. Calibration alone is often viewed as too weak to support informative decisions. Our results show that, when multiple calibrated forecasts are observed jointly, calibration can nevertheless impose enough structure to support tractable and learnable aggregation rules that outperform standard best-expert benchmarks.

\subsection{Related Work}

\paragraph{Robust Forecast Aggregation}
Our work contributes to the literature on robust forecast aggregation, which studies how to combine multiple forecasts into a single forecast when the aggregator does not know the underlying information structure. The most direct predecessor is \citet{arieli2018robust}, who introduce a regret-based criterion for evaluating aggregation schemes relative to the Bayesian benchmark and construct low-regret aggregation schemes in two-expert settings with either Blackwell-ordered or conditionally independent signals. \citet{guo2025algorithmic} extend this line of work by developing an algorithmic framework for computing robust aggregation rules, with a primary focus on conditionally independent information structures. \citet{de2023robust} study robust aggregation of correlated information from a decision-theoretic perspective and characterize robustly optimal strategies in binary-state, binary-action problems. In a similar spirit, \citet{levy2022combining} study forecast combination when each expert's signal is understood individually but the correlation structure across experts is ambiguous.
\citet{neyman2022you} identify conditions under which simple aggregation methods obtain nontrivial guarantees, and \citet{frongillo2025robust} show that allowing the principal to ask additional structured queries can overcome some impossibility results in robust forecast aggregation.
More recently, \citet{chen2026prior} study \emph{prior-agnostic} robust aggregation, in which the aggregator observes only the experts' reports and knows neither the prior nor the underlying state space. A related strand relaxes the assumption that experts report truthful Bayesian posteriors, studying robust aggregation with adversarial experts \citep{guo2024robust} and with base-rate-neglecting experts \citep{kong2024surprising}.

Much of this literature studies Bayesian experts under specific restrictions on the information structures, priors, dependence, or elicitation protocols. In contrast, our work takes calibrated forecasters as primitives and studies the class of information structures induced by calibration. This leads to a different robustness problem: rather than assuming a particular dependence structure among signals, we ask how forecasts can be aggregated when calibration restricts but does not fully identify the underlying information structure. We further extend the analysis to online settings and design no-regret aggregation algorithms for this calibrated-forecaster benchmark.

\paragraph{Online Learning and Regrets} 
Our results are related to the literature on benchmarks and no-regret algorithms in online decision-making problem. The standard benchmark is the optimal-in-hindsight benchmark, which leads to the external regret, where the learner competes with the best fixed action, expert, or decision rule chosen in hindsight. No-external-regret guarantees can be achieved by Hedge \citep{freund1997decision}, follow-the-regularized-leader \citep{shalev2025online}, and online mirror descent \citep{zinkevich2003online} under full-gradient feedback, by EXP3 algorithms under adversarial bandit feedback \citep{auer2002nonstochastic}, and by EXP4 in contextual bandit problems \citep{auer2002nonstochastic,zhou2015survey}.

Another benchmark is swap regret, which compares against arbitrary action mappings. Internal regret concerns pairwise action switches, while swap regret permits arbitrary action mappings. These stronger regrets are closely connected to calibration and correlated equilibrium \citep{foster1997calibrated}, and can be achieved by regret-matching algorithms \citep{hart2000simple} or by reductions from external-regret minimization \citep{blum2007external}. 

Our benchmark is the robustly optimal aggregation rule subject to the information structures induced by calibration, and is weakly stronger than the optimal-in-hindsight benchmark. More importantly, our setting differs from standard online learning because the learner observes only forecasts. With forecast-only feedback, payoff-based learning algorithms cannot directly evaluate experts' reports or feasible actions, and therefore are not directly applicable. Our online algorithm is designed for this information constraint and attains no-regret guarantees with respect to the robust aggregation benchmark.

\paragraph{Calibration and Decision Guarantees}

Our work is also related to the literature on calibration and its improvements. The classical calibration literature begins with \citet{dawid1982well}'s notion of calibrated probability assessment; \citet{oakes1985self} shows the limitations of deterministic self-calibration, whereas \citet{foster1998asymptotic} establish that asymptotic calibration can be achieved through randomized forecasting procedures.
However, \citet{gneiting2007probabilistic} emphasize that calibration alone is not sufficient for useful prediction---sharpness, or informativeness, is also needed, which motivated refinements or improvements of calibration. 

Multicalibration strengthens ordinary calibration by requiring calibration not only conditional on the predictor's own forecast, but also over a rich collection of identifiable subpopulations \citep{hebert2018multicalibration}. Omniprediction builds on this idea from a decision-making perspective \citep{gopalan2022omnipredictors}, and recent work shows a close connection between multicalibration and omniprediction \citep{gopalan2023swap,garg2024oracle}. Our work shares with these two notions the view that calibration should be understood through its implications for downstream decision making. However, the object of study is different. Multicalibration and omniprediction primarily concern the construction of a single predictor whose forecasts are sufficiently informative for many downstream losses. In contrast, we take multiple calibrated forecasters as primitives and study how their reports can be robustly aggregated when the joint information structure is unknown. 
A related line of work studies calibration directly through its downstream decision value, including $U$-calibration \citep{kleinberg2023ucalibration} and swap-regret guarantees for downstream decision tasks \citep{hu2024predict}. These works concern how to produce or evaluate a single predictor so that downstream agents act well, whereas we study how a decision-maker aggregates several already-calibrated reports under an unknown information structure.

\paragraph{Statistical Forecast Pooling}
A more closely related line of work studies model-based aggregation of multiple probability forecasts. \citet{ranjan2010combining} show that linear pools of calibrated forecasts are generally uncalibrated and lack sharpness, highlighting that aggregating calibrated forecasts is itself a \textit{nontrivial} problem. \citet{satopaa2014combining} propose a simple logit model for combining expert forecasts and demonstrate strong empirical performance on geopolitical forecasting data. \citet{ernst2016bayesian} study forecast aggregation in the partial information framework and derive an explicit Gaussian aggregator for a one-shot problem with two forecasters. These works either impose specific statistical models on the forecast-generating process, focus on the two-forecaster case, or primarily provide empirical aggregation methods. In contrast, our work takes calibrated forecasters as primitives and studies robust aggregation over the class of information structures induced by calibration, without assuming a parametric model or restricting attention to two experts.

\paragraph{Notation} Given a positive integer $n$, let $[n]$ denote the set $\{1,2, \cdots, n\}$. Given a vector $v\in \mathbb{R}^n$, let $[v]_i$ denote its $i$-th element.  Given a finite set $S$, let $|S|$ denote the number of its elements. Given a sequence of vectors $v_1, v_2, \cdots, v_t$, we denote it by $v_{1:t}$. 

\section{Model}
\label{sec:model}

Consider a \textit{dynamic} decision-making problem $\Gamma = (\Omega, \mathcal{A}, u, T)$
with a binary state space $\Omega=\{0,1\}$ of nature, a finite action set $\mathcal{A}$ of size $|\mathcal{A}|=d$, a bounded payoff function $u:\mathcal{A}\times \Omega\to [-U, U]$, and a finite horizon $T\in\mathbb{N}_+$. 
At each day $t$, the decision-maker observes the available forecasts, chooses an action $a_t\in\mathcal{A}$, and the state $\omega_t\in\Omega$ is then revealed or withheld depending on the feedback model. 
The decision maker's payoff given any sequence of action-state pairs $\{(a_t,\omega_t)\}_{t=1}^T$ is 
\begin{align*}
    \sum_{t=1}^T u(a_t,\omega_t).
\end{align*}

\paragraph{Expert Forecasts}
The decision-maker has access to $n$ experts, each of whom makes forecasts about the probability of the event that the state of nature equals $1$, i.e., $\Pr(\omega=1)$. 
Let $e_t^i$ be the forecast made by expert $i$ at day $t$, and let 
$e_t = (e_t^1, \dots, e_t^n)$ be the entire forecast profile. 
Let $E^i$ be the collection of distinct reports made by expert $i$ over the time horizon $T$, and let $E$ denote their Cartesian product, i.e., $E = E^1\times E^2\times \cdots \times E^n$. Denote the size of $E$ by $m$. In online sections, we take these finite forecast alphabets as known before play.
In our model, we assume that all experts satisfy $\varepsilon_T$-calibration, where $\varepsilon_T$ represents the calibration error. 
That is, for any expert $i\in[n]$ and any $e^i\in E^i$, 
we have
\begin{equation}\label{eq_calib_deter}
    \left\vert \frac{\sum_{t\leq T} \indicate{\omega_t = 1}\cdot \indicate{e^i_t = e^i}}{\sum_{t\leq T} \indicate{e^i_t = e^i}} - e^i\right\vert \le \varepsilon_T. 
\end{equation}

\paragraph{Online Decisions and Regrets}
In this online decision problem, the decision maker chooses an action in each period based on the historical observations. In our model, we consider two types of feedback:
\begin{itemize}
    \item forecast-only feedback: at any time $t\leq T$, the decision maker only observes the forecasts from all experts from time $1$ to $t$.
    \item state feedback: at any time $t\leq T$, the decision maker observes the forecasts from all experts from time $1$ to $t$, and the realized state from time $1$ to $t-1$.
\end{itemize}

To measure the performance of the decision maker, there are two natural benchmarks (defined informally below). 
\begin{itemize}
    \item Optimal-in-hindsight: The optimal-in-hindsight benchmark is a standard and widely used benchmark, which is defined as the payoff obtained by the best report-to-action policy based on a single expert's forecasts and selecting, in hindsight, the expert whose reports yield the highest payoff.
    \begin{definition}[Optimal-in-Hindsight Benchmark]\label{def_OIH_benchmark}
        Let $U^i$ denote the optimal value achieved when the decision-maker uses only the reports from expert $i$ \footnote{Under exact calibration, the value of the best single-expert policy is determined by the expert's forecast sequence. With approximate calibration, different calibrated state sequences can lead to values that differ by the calibration error term.}, namely,
        \begin{equation*}
            U^i = \max_{\sigma:E^i\to \mathcal{A}} \sum_{t\le T} u(\sigma(e_t^i), \omega_t),
        \end{equation*}
        where $\omega_1,\ldots,\omega_T$ is the realized state sequence. Then, the optimal-in-hindsight benchmark is defined as
        \begin{equation}
            \oih := \max_{i\in [n]} U^i.
        \end{equation}
    \end{definition}

    \item Robust-maxmin: The robust-maxmin benchmark is the maximal payoff that can be guaranteed by a mixed strategy over action sequence space against all calibration-compatible state sequence.

    \begin{definition}[Robust-Maxmin Benchmark]\label{def_RMM_deter}
        Let $\Omega(e_{1:T})$ be the set of state sequences that satisfy the calibration constraints \eqref{eq_calib_deter} for all experts. Then, the robust-maxmin benchmark is defined as 
        \begin{equation}\label{eq_origin_def_RMM}
            \robust :=
            \max_{x\in (\Delta(\mathcal{A}))^m}
            \min_{\omega_{1:T}\in \Omega(e_{1:T})}
            \mathbb{E}\left[
                \sum_{t\leq T} u(a_t,\omega_t)
            \right],
        \end{equation}
        where $a_t\sim x(e_t)$ with $x(e_t)\in\Delta(\mathcal{A})$, and the expectation is taken only over this randomization. 
        
    \end{definition}
\end{itemize}

In our paper, we focus primarily on the robust maxmin benchmark. Given this benchmark, the regret of the decision maker given any realized sequence of actions $a^*_1,\dots,a^*_T$ is 
\begin{align*}
    \regret = \robust - \min_{\omega_{1:T}\in \Omega(e_{1:T})} \sum_{t=1}^T u(a^*_t,\omega_t).
\end{align*}

\paragraph{Bayesian Environments}
In Bayesian environments, we make the additional assumption that both the states and the forecasts are generated via a stochastic process. 
Specifically, on each day $t\leq T$, the state $\omega_t$ is drawn independently according to a prior distribution $q$ over $\Omega$. 
By slightly overloading the notation, we also use $q$ to denote the probability of the state being $1$ under distribution $q$. 
Then, each expert $i$ receives a signal $s^i_t$ according to the information structure $\pi^i : \Omega \to \Delta(S^i)$, where $S^i$ is a finite signal space. 
In Bayesian environments, we assume that the experts' signals $s^1_t,\dots,s^n_t$ are drawn independently across experts conditional on the realized state $\omega_t$.
Let $\pi^i(s;\omega)$ be the probability that the signal is $s$ when the state is $\omega$. 
In Bayesian environments, each expert's forecast is their posterior belief over the state conditional on receiving the signal $s^i_t$, which, by Bayes' rule, is 
\begin{align*}
e^i_t = \frac{q\cdot \pi^i(s^i_t;1)}{q\cdot \pi^i(s^i_t;1) + (1-q)\cdot \pi^i(s^i_t;0)}.
\end{align*}
In this case, the prior $q$, together with the experts' forecast profile $e$, is sufficient to determine the posterior probability that the state is $1$ conditional on observing $e$.

\begin{definition}[Bayesian Benchmark]
    The Bayesian benchmark is the maximal expected payoff attainable by strategies that condition on the realized forecast profile:
    \begin{equation}\label{eq_def_baye_benc}
        \operatorname{Bay}
        :=
        \max_{\sigma: E\to \mathcal{A}}
        \mathbb E\left[
            \sum_{t\leq T} u(\sigma(e_t), \omega_t
            )
        \right],
    \end{equation}
    where the expectation is taken under the Bayesian posterior induced by the observed forecast profile.
\end{definition}

\section{Offline Benchmarks}\label{sec:offline}
In this section, we study the offline setting, where the entire sequence of forecast profiles over $T$ days is fixed. We first introduce the formal definition of the robust-maxmin benchmark and analyze its computational complexity. We then show that this benchmark is stronger than the optimal-in-hindsight benchmark, up to the calibration error term. Finally, in the Bayesian setting, we discuss the relationship between the robust-maxmin benchmark and the optimal Bayesian benchmark.

\subsection{Robust-Maxmin Benchmark} 

Although each expert is calibrated, relying on the forecasts of a single expert may discard decision-relevant information contained in the joint forecast profile of all experts. The key idea behind the robust-maxmin benchmark is therefore to aggregate experts’ forecasts by exploiting the implications of calibration. We next show how the calibration condition can be used to construct a more informative forecast of the state of nature.

\paragraph{Calibration Constraint} As illustrated by the example in the introduction, aggregating forecasts naturally involves partitioning the time horizon into groups according to the experts' forecast profiles and assigning a posterior probability to each group. Although the example is simple, it captures the intrinsic process of forecast aggregation.

More precisely, for each forecast profile $e$, we group together all days on which $e_t=e$. The calibration constraint then restricts the empirical frequency with which the state $\omega_t=1$ occurs within each such group. Let $\rho:E\to[0,1]$ denote the mapping that assigns to each forecast profile $e$ this conditional empirical frequency. That is, for each $e\in E$ with positive empirical frequency,
\begin{equation*}
    \rho(e) =
    \frac{\sum_{t\leq T} \indicate{e_t = e}\indicate{\omega_t=1}}
    {\sum_{t\leq T}\indicate{e_t=e}}.
\end{equation*}

In this notation, the calibration condition can be expressed as a set of linear constraints on the vector \(\rho :=  (\rho(e))_{e\in E}\):
\begin{equation*}
    \left\vert
    \frac{\sum_{e:[e]_i = e^i}\sum_{t\le T}\indicate{e_t = e} \rho(e)}
    {\sum_{t\leq T}\indicate{e_t^i = e^i}}
    - e^i
    \right\vert
    \le \varepsilon_T,
    \quad \forall e^i\in E^i,\ i\in [n].
\end{equation*}

Let \(\edp_T\) denote the empirical distribution over forecast profiles induced by the observed sequence over \(T\) days. That is, $\edp_T(e)$ is the empirical frequency of profile $e$. For each expert $i$ and forecast $e^i\in E^i$, we write $\edp_T(e^i) := \sum_{e:[e]_i=e^i} \edp_T(e)$ for the empirical frequency with which expert $i$ reports forecast $e^i$. In terms of \(\edp_T\), the calibration constraint is equivalently written as
\begin{equation*}
    \left\vert
    \frac{\sum_{e:[e]_i = e^i}\edp_T(e) \rho(e)}{\edp_T(e^i)}
    - e^i
    \right\vert
    \le \varepsilon_T,
    \quad \forall e^i\in E^i,\ i\in[n].
\end{equation*}
We use the following form, which is equivalent whenever $\edp_T(e^i)>0$ and remains well-defined when $\edp_T(e^i)=0$:
\begin{equation}\label{eq_calibration_consraint_def}
    \left\vert
    \sum_{e:[e]_i = e^i}\edp_T(e)\bigl(\rho(e)-e^i\bigr)
    \right\vert
    \le \edp_T(e^i)\varepsilon_T\,,
    \quad \forall e^i\in E^i,\ i\in[n].
\end{equation}

The calibration constraint generally does not pin down a unique $\rho$, even when $\varepsilon_T = 0$. Let $\mathcal C(\edp_T)$ denote the set of all mappings $\rho$ that satisfy the calibration constraints~\eqref{eq_calibration_consraint_def} induced by the empirical distribution $\edp_T$. 

\paragraph{Robust-Maxmin Benchmark}
When choosing actions, the decision-maker must account for all mappings $\rho \in \mathcal{C}(\edp_T)$ and perform well uniformly over this set. A natural criterion is therefore to maximize the worst-case expected payoff. We refer to this criterion as the \textit{robust aggregation} of calibrated forecasters.

Given a mapping $\rho:E\to [0, 1]$, define the associated payoff mapping $v_\rho: E\to [-U, U]^d$ by
\begin{equation}
    v_\rho(e) = \rho(e) u_1 + (1-\rho(e))u_0,
\end{equation}
where $u_1:=(u(a, 1))_{a\in \mathcal{A}}$ and $u_0:=(u(a, 0))_{a\in \mathcal{A}}$. Then, the robust-maxmin benchmark is the maximal payoff that can be guaranteed by a decision rule against all calibration-compatible posterior mappings. By partitioning the time horizon according to the realized forecast profiles and using the mapping $\rho$, the robust-maxmin benchmark can be written as
\begin{equation}\label{eq_expected_format_RMM}
    \robust =
    \max_{x\in (\Delta(\mathcal{A}))^m}
    \min_{\rho\in \mathcal{C}(\edp_T)}
    T\sum_{e\in E} \edp_T(e)\, x(e)^\top v_\rho(e).
\end{equation}

    

The calibration constraint imposes correlation into $\rho(e)$ between different forecast profiles. However, in the definition of the robust-maxmin benchmark, the decision-maker's strategy $x$ belongs to the space $\bigl(\Delta(\mathcal{A})\bigr)^m$ rather than $\Delta(\mathcal{A}^m)$ that allows correlations. The following result shows that allowing correlation across forecast profiles does not improve the robust-maxmin value. The proof compares correlated strategies and profile-wise mixed strategies through their induced marginal mixed actions, showing that each strategy in one formulation has a counterpart in the other with the same payoff against every feasible posterior mapping.

\begin{proposition}\label{thm_indepen_strate_optimal}
    Independent mixed strategies achieve the same robust-maxmin value as correlated strategies. Formally, 
    let $\alpha= (\alpha(e))_{e\in E}$ denote a pure strategy where $\alpha(e)\in \mathcal{A}$. Then, \begin{equation}\label{eq_equiv_cor_indep}
        \max_{x\in (\Delta(\mathcal{A}))^m}\min_{\rho\in \mathcal{C}(\edp_T)} \sum_{e\in E} \edp_T(e)x(e)^\top v_\rho(e) = \max_{y\in \Delta(\mathcal{A}^m)}\min_{\rho\in \mathcal{C}(\edp_T)} \sum_{\alpha\in \mathcal{A}^m} y(\alpha) \sum_{e\in E} \edp_T(e)\mathbf{1}_{\alpha(e)}^\top v_\rho(e),
    \end{equation}
    where $\mathbf 1_a$ denotes the unit vector corresponding to action $a$.
\end{proposition}
\begin{proof}
    See Appendix~\ref{ap:proofs_of_section_offline}.
\end{proof}

Next, we analyze the computational complexity of the robust-maxmin benchmark. The benchmark is the value of a saddle-point problem, where the objective is bilinear in the decision rule $x$ and the mapping $\rho$, while the feasible sets for both variables are described by linear constraints. As shown in the proof of the following theorem, after fixing $x$, dualizing the inner minimization problem yields an equivalent linear program over the decision rule $x$ and the associated dual variables. Hence, the robust-maxmin benchmark can be computed in polynomial time.

\begin{theorem}\label{prop_compute_RMM}
    Let $r:=\sum_{i\in[n]} |E^i|.$ Then the robust-maxmin benchmark can be computed by solving a linear program with $md+m+2r$ variables and $md+3m+2r$ constraints. In particular, it can be computed in time polynomial in $m$, $d$, and $r$ by standard linear-programming algorithms \footnote{The parameter $m=|E| = \prod_i |E^i|$ is the explicitly represented profile-space size and may grow exponentially with the number of experts.}
\end{theorem}
\begin{proof}
    See Appendix~\ref{ap:proofs_of_section_offline}.
\end{proof}

\subsection{Compared with the Optimal-in-Hindsight Benchmark}



We next compare the robust-maxmin benchmark with the optimal-in-hindsight benchmark defined in the Definition~\ref{def_OIH_benchmark}. The following result shows that, when experts are perfectly calibrated, the robust-maxmin benchmark cannot perform worse than the optimal-in-hindsight benchmark. Under approximate calibration, the same comparison holds up to an error term proportional to the calibration error. The intuition is that robust aggregation can always replicate the best single expert’s policy, while approximate calibration ensures that the worst calibrated belief cannot distort the payoff of this policy by more than the calibration-error term.

\begin{proposition}\label{prop:robust_vs_oih}
    Suppose that all experts are $\varepsilon_T$-calibrated. Then, 
    \begin{equation*}
        \robust \geq \oih - 4\varepsilon_TUT.
    \end{equation*}
    In particular, if all experts are perfectly calibrated, i.e., $\varepsilon_T=0$, then $\robust \geq \oih$.
\end{proposition}
\begin{proof}
    See Appendix~\ref{ap:proofs_of_section_offline}.
\end{proof}

\begin{remark}
    The key difference between the robust-maxmin and the optimal-in-hindsight benchmark lies in the granularity of the information used to form decisions. The optimal-in-hindsight benchmark commits to a single expert and therefore partitions the time horizon only according to that expert's forecasts. In contrast, the robust-maxmin benchmark uses the entire forecast profile of all experts. It partitions the time horizon into finer groups, one for each joint forecast profile, and then uses the calibration constraints to restrict the possible empirical frequency of the state within each group. Thus, the robust-maxmin benchmark exploits decision-relevant information contained in the joint forecasts that may be discarded when using only any single expert. 
\end{remark}

\subsection{Compared with the Bayesian Benchmark}

In a Bayesian environment, each expert observes a private signal and reports the posterior probability of the state conditional on that signal. Without loss of generality, we assume that signals and the induced posterior forecasts are in one-to-one correspondence.

In this case, the prior $q$, together with the experts' forecast profile $e$, is sufficient to determine the posterior probability that the state is $1$ conditional on observing $e$. By Bayes' rule,
\begin{equation*}
    \operatorname{Pr}(\omega=1\mid e)
    =
    \frac{(1-q)^{n-1}\prod_{i\in[n]} e^i}
    {(1-q)^{n-1}\prod_{i\in[n]} e^i
    +
    q^{n-1}\prod_{i\in[n]}(1-e^i)}.
\end{equation*}
For notational convenience, define $\rho^\ast:E\to[0,1]$ by setting $\rho^\ast(e)$ to be the Bayesian posterior probability associated with forecast profile $e$.

In the Bayesian environment, the realized forecast profile at each day can be regarded as being drawn from a distribution over $E$, denoted by $\mathbf{P}\in\Delta(E)$. The Bayesian benchmark is obtained by evaluating the decision-maker's expected payoff under this distribution and, for each realized forecast profile, choosing a strategy that maximizes the payoff evaluated at the corresponding posterior belief, that is,
\begin{equation*}
    \operatorname{Bay}
    =
    \max_{x\in(\Delta(\mathcal A))^m}
    T\sum_{e\in E}
    \mathbf P(e)\,x(e)^\top v_{\rho^\ast}(e).
\end{equation*}


In Bayesian environments, the Bayesian benchmark represents the optimal expected reward attainable by a decision-maker who knows the prior and correctly updates beliefs using the experts' forecast profile. It is therefore the natural performance benchmark in such environments. 
This naturally raises the following question:
\begin{center}
    Does the robust-maxmin benchmark match this optimal benchmark?
\end{center}

To make this comparison precise, we adapt the robust-maxmin benchmark to the Bayesian environment by replacing the empirical distribution induced by the realized sample path with the underlying distribution of forecast profiles. Define the calibration constraint using $\mathbf P$ by
\begin{equation}\label{eq_cali_true_prob}
    \left\vert
    \sum_{e:[e]_i=e^i}\mathbf P(e)\bigl(\rho(e)-e^i\bigr)
    \right\vert
    \le
    \varepsilon_T\sum_{e:[e]_i=e^i}\mathbf P(e),
    \quad \forall e^i\in E^i,\ i\in[n].
\end{equation}
Let $\mathcal{C}(\mathbf{P})$ denote the set of all $\rho$ that satisfy the calibration constraints induced by $\mathbf{P}$. The corresponding robust-maxmin benchmark is then defined as
\begin{equation}\label{eq_def_RMM_stoc}
    \robust := \max_{x\in (\Delta(\mathcal{A}))^m}\min_{\rho\in \mathcal{C}(\mathbf{P})} T\sum_{e\in E} \mathbf{P}(e)x(e)^\top v_\rho(e). 
\end{equation}

We first record a direct comparison between the robust-maxmin benchmark and the Bayesian benchmark.
\begin{proposition}\label{prop_rmm_less_than_bay}
    The robust-maxmin benchmark cannot exceed the Bayesian benchmark, i.e., 
    \begin{equation*}
        \robust \leq \operatorname{Bay}
    \end{equation*}
\end{proposition}

\begin{proof}
    Since the Bayesian posterior $\rho^\ast$ satisfies the calibration constraint~\eqref{eq_cali_true_prob}, we have $\rho^\ast\in\mathcal C(\mathbf P)$. Hence, by Sion's minimax theorem \citep{sion1958general}, 
    \begin{align*}
        \robust
        =
        \min_{\rho\in \mathcal C(\mathbf P)}
        \max_{x\in(\Delta(\mathcal A))^m}
        T\sum_{e\in E}\mathbf P(e)x(e)^\top v_\rho(e) 
        \leq
        \max_{x\in(\Delta(\mathcal A))^m}
        T\sum_{e\in E}\mathbf P(e)x(e)^\top v_{\rho^\ast}(e)
        =
        \operatorname{Bay}.
    \end{align*}
\end{proof}

The following example shows that this inequality can be strict. 
    
\begin{example}\label{ex_cond_indep}
    Consider an environment with two experts. Suppose that both experts are perfectly calibrated, i.e., $\varepsilon_T=0$. The decision-maker has two actions, $\mathcal A=\{a_1,a_2\}$, and the payoff function is given by
    \[
        u(a_1,1)=0,\qquad
        u(a_2,1)=3,\qquad
        u(a_1,0)=1,\qquad
        u(a_2,0)=0.
    \]
    The prior probability of state $1$ is $q=\frac12$.

    Expert $1$ has signal set $\{s_1^1,s_1^2\}$, with information structure
    \begin{gather*}
        \pi^1(s_1^1;1)=\frac13,\qquad \pi^1(s_1^1;0)=\frac23,\\
        \pi^1(s_1^2;1)=\frac23,\qquad \pi^1(s_1^2;0)=\frac13.
    \end{gather*}
    The induced forecasts are $E^1=\left\{\frac13,\frac23\right\}.$ Expert $2$ has signal set $\{s_2^1,s_2^2\}$, with information structure
    \begin{gather*}
        \pi^2(s_2^1;1)=\frac17,\qquad \pi^2(s_2^1;0)=\frac57,\\
        \pi^2(s_2^2;1)=\frac67,\qquad \pi^2(s_2^2;0)=\frac27.
    \end{gather*}
    The induced forecasts are $E^2=\left\{\frac16,\frac34\right\}.$ The resulting distribution over forecast profiles is
    \begin{gather*}
        \mathbf{P}\bigl((\frac{1}{3}, \frac{1}{6})\bigr) = \frac{11}{42},\qquad \mathbf{P}\bigl((\frac{1}{3}, \frac{3}{4})\bigr) = \frac{5}{21},\\ \mathbf{P}\bigl((\frac{2}{3}, \frac{1}{6})\bigr) = \frac{1}{6},\qquad \mathbf{P}\bigl((\frac{2}{3}, \frac{3}{4})\bigr) = \frac{1}{3}.
    \end{gather*}

    In this environment, the Bayesian benchmark is $\frac{5}{3}T,$ whereas the robust-maxmin benchmark is $\frac{23}{14}T.$ Hence, the Bayesian benchmark is strictly larger than the robust-maxmin benchmark. 
    
\end{example}

The example above can be understood from the perspective of the nested max-min structure of the robust-maxmin benchmark. Fix any $\rho\in\mathcal C(\mathbf P)$ and define the posterior error mapping $\beta: E\to [-1, 1]$ by $\beta(e):=\rho(e)-\rho^\ast(e)$ for all $e\in E.$ Since the Bayesian posterior $\rho^\ast$ also satisfies the calibration constraints, the mapping $\beta$ must satisfy
\begin{equation}\label{eq_const_beta}
    \sum_{e:[e]_i=e^i}\mathbf P(e)\beta(e)=0,
    \qquad \forall e^i\in E^i,\ i\in[n].
\end{equation}
Thus, the feasible posterior error mappings lie in the intersection of a linear subspace, determined by~\eqref{eq_const_beta}, with the box constraints 
\[
    -\rho^\ast(e)\leq \beta(e)\leq 1-\rho^\ast(e),
    \qquad e\in E.
\]

Substituting $\rho=\rho^\ast+\beta$ into the robust-maxmin benchmark gives
\begin{align*}
    \robust
    &=
    \max_{x\in(\Delta(\mathcal A))^m}
    \min_{\beta}
    T\sum_{e\in E}\mathbf P(e)\,
    x(e)^\top
    \Bigl(
        \bigl(\rho^\ast(e)+\beta(e)\bigr)u_1
        +
        \bigl(1-\rho^\ast(e)-\beta(e)\bigr)u_0
    \Bigr)\\
    &=
    \max_{x\in(\Delta(\mathcal A))^m}
    \left\{
    T\sum_{e\in E}\mathbf P(e)\,x(e)^\top v_{\rho^\ast}(e)
    +
    \min_{\beta}
    T\sum_{e\in E}\mathbf P(e)\,\beta(e)\,
    x(e)^\top (u_1-u_0)
    \right\},
\end{align*}
where the minimization is over all feasible posterior error mappings satisfying~\eqref{eq_const_beta} and the box constraints above. The first term is the Bayesian payoff of the fixed decision rule $x$ under the true Bayesian posterior; maximizing this term over $x$ gives the Bayesian benchmark. However, in the robust-maxmin benchmark, the decision rule $x$ is chosen while anticipating the worst feasible posterior error in the second term. This inner minimization can change the optimal choice of $x$ and reduce the resulting value relative to the Bayesian benchmark. Therefore, calibration alone does not generally allow the robust-maxmin benchmark to match the Bayesian benchmark; in general, the robust-maxmin benchmark can be strictly smaller. 

\section{Online Forecast-only Feedback Model}\label{sec:online_forecast_feedback}

In this section, we study online decision-making with forecast-only feedback and develop an algorithm that achieves no regret with respect to the robust-maxmin benchmark. This result shows that the decision-maker can compete against a benchmark that is more compelling than the optimal-in-hindsight benchmark by observing all calibrated experts' forecasts and exploiting the decision-relevant information. 

\subsection{Online Protocol and Robust Regret}

We impose the following additional assumption on the generation of experts' forecasts. The necessity of this assumption is illustrated at the end of this section by an example showing that, without it, no online algorithm can guarantee sublinear robust regret uniformly over arbitrary deterministic sequences of forecast profiles.

\begin{assumption}\label{as:fixed_distribution}
    At each day, experts' forecast profiles are sampled independently from a fixed distribution $\mathbf{P}\in\Delta(E)$, while the state sequence remains non-stochastic.
\end{assumption}

This assumption is motivated by two observations. First, in the motivating example from the introduction, it is plausible that no algorithm can achieve a no-regret guarantee with respect to the robust-maxmin benchmark in~\eqref{eq_expected_format_RMM}: during the first third of the horizon, algorithms are likely to choose the same strategy, which may result in regret of order $T$. Second, the calibration constraints are fully determined by the empirical distribution $\edp_T$ of forecast profiles. Consequently, any two sequences of forecast profiles with the same empirical distribution $\edp_T$ induce the same robust-maxmin benchmark. However, the order in which forecast profiles appear can affect the algorithm's decisions, making a no-regret guarantee difficult to attain without further restrictions. 

Now, we specify the decision-making protocol under this assumption. On each day $t\leq T$, the decision-maker observes the past forecast profiles $e_{1:t-1}$ and the current forecast profile $e_t$, drawn from $\mathbf P$. Based on this information, the decision-maker samples an action $a_t\in\mathcal{A}$ according to a mixed strategy, denoted by $x_t(e_t)\in \Delta(\mathcal{A})$. This formulation includes deterministic choices as a special case. 

The \textit{robust regret} is defined relative to the robust-maxmin benchmark as follows. 

\begin{definition}[Robust Regret]\label{Robust Regret}
    Suppose that Assumption~\ref{as:fixed_distribution} holds and that $\mathcal C(\mathbf P)$ is nonempty. Given an action sequence $\{a_t\}_{t=1}^T$, its robust regret is defined as
    \begin{equation}\label{eq_robust_regret}
        \regret
        :=
        \robust
        -
        \min_{\rho\in\mathcal C(\mathbf P)}
        \mathbb E\!\left[
            \sum_{t=1}^T  u(a_t, \omega_t)
        \right] = \robust
        -
        \min_{\rho\in\mathcal C(\mathbf P)}
        \mathbb E\!\left[
            \sum_{t=1}^T  x_t(e_t)^\top v_\rho(e_t)
        \right],
    \end{equation}
    where $\robust$ is defined in~\eqref{eq_def_RMM_stoc}, $\mathcal C(\mathbf P)$ is defined by the calibration constraints in~\eqref{eq_cali_true_prob}, and the expectation is taken over the sampled forecast sequence and the learner's randomization. 
\end{definition}

Here, we evaluate the sequence of actions $\{a_t\}_{t=1}^T$ by its worst-case expected payoff over all $\rho\in\mathcal{C}(\mathbf{P})$. This criterion is natural in the forecast-only feedback setting because the decision-maker neither observes the realized state sequence nor receives payoff feedback over the horizon, and the robust-maxmin benchmark is defined using the same worst-case evaluation. The theorem below therefore competes with $\robust$ based on the distribution $\mathbf P$, not with the realized empirical benchmark based on $\edp_T$.

An online algorithm is called \textit{no-robust-regret} if the sequence of actions $\{a_t\}_{t=1}^T$ generated by the algorithm incurs sublinear robust regret, i.e., $\regret=o(T)$. Before presenting our no-robust-regret learning algorithm, we first rewrite the robust regret in a more convenient form using conditional expectations. 

\paragraph{Filtration.}
Let $\mathcal F_0$ be the trivial $\sigma$-algebra, and for each $t\geq 1$, let $\mathcal{F}_{t-1}:=\sigma(e_1,\ldots,e_{t-1})$ denote the $\sigma$-algebra generated by the historical forecast profiles. Since the forecast profiles are sampled independently across time, $e_t$ is independent of $\mathcal F_{t-1}$ and, by taking conditional expectations day by day, the expected realized payoff can be expanded into the corresponding expectation under \(\mathbf P\) at each day.

\begin{lemma}[Iterated Expectation Expansion]\label{lem:iterated_expectation_full} 
If the mixed strategy $x_t(e_t)$ is induced by a decision rule $x_t:E\to\Delta(\mathcal A)$ that is $\mathcal F_{t-1}$-measurable for each $t\leq T$, then, for any $\rho\in\mathcal C(\mathbf P)$,
\[
    \mathbb E\!\left[\sum_{t=1}^T x_t(e_t)^\top v_\rho(e_t)\right]
    =
    \sum_{t=1}^T
    \mathbb E\!\left[
        \sum_{e\in E}\mathbf P(e)\,x_t(e)^\top v_\rho(e)
    \right].
\]
Consequently,
\begin{equation}\label{eq:regret_equiv_form_full}
    \regret
    =
    \robust
    -
    \min_{\rho\in\mathcal C(\mathbf P)}
    \sum_{t=1}^T
    \mathbb E\!\left[
        \sum_{e\in E}\mathbf P(e)\,x_t(e)^\top v_\rho(e)
    \right].
\end{equation}
\end{lemma}
\begin{proof}
    See Appendix~\ref{ap:proof_of_section_forecast_feecback}.
\end{proof}

\subsection{Plug-in Robust Aggregation: A No-Robust-Regret Algorithm}

In the online setting, the main difficulty is that the distribution $\mathbf P$ is unknown. Consequently, the robustly optimal decision rule associated with the robust-maxmin benchmark cannot be computed exactly during learning. A natural approach is to estimate $\mathbf P$ from the observed forecast profiles and compute a robustly optimal decision rule under the estimated distribution.

Define the empirical distribution by
\begin{equation}\label{eq:empirical_dist}
    \hat{\mathbf P}_{t-1}(e)
    :=
    \begin{cases}
        0, & t=1,\\ 
        \frac{\sum_{\tau=1}^{t-1}\indicate{e_\tau=e}}{t-1}, & t\ge 2,
    \end{cases}
    \qquad \forall e\in E.
\end{equation}
Given $\hat{\mathbf P}_{t-1}$, the corresponding empirical calibration set is
\begin{equation*}
    \mathcal C(\hat{\mathbf P}_{t-1})
    :=
    \left\{
    \rho\in [0,1]^m:\ 
    \left|
    \sum_{e:\,[e]_i=e^i}
    \hat{\mathbf P}_{t-1}(e)\bigl(\rho(e)-e^i\bigr)
    \right|
    \leq
    \sum_{e:\,[e]_i=e^i}\hat{\mathbf{P}}_{t-1}(e)\varepsilon_T,
    \quad
    \forall e^i\in E^i,\ i\in[n]
    \right\}.
\end{equation*}
This non-normalized form avoids division by empirical frequencies that may be zero.

However, due to finite-sample fluctuations, the empirical calibration constraints need not be jointly feasible. In other words, $\mathcal C(\hat{\mathbf P}_{t-1})$ may be empty, in which case the plug-in saddle-point problem is not well defined. To avoid this issue, we relax the hard calibration constraints and instead penalize calibration violations directly in the objective. Specifically, at day $t$, the decision-maker solves the following penalized saddle-point problem:
\begin{equation}\label{eq:penalized_SP}
    \max_{x\in(\Delta(\mathcal A))^m}
    \min_{\rho\in[0,1]^m}
    \left\{
    \sum_{e\in E}\hat{\mathbf P}_{t-1}(e)\,x(e)^\top v_\rho(e)
    +
    \gamma
    \sum_{i\in[n]}\sum_{e^i\in E^i}
    \left|
    \sum_{e:\,[e]_i=e^i}
    \hat{\mathbf P}_{t-1}(e)(\rho(e)
    - e^i)
    \right|
    \right\},
\end{equation}
where $\gamma>0$ is a penalty coefficient. We refer to the resulting online procedure as the \emph{Plug-in Robust Aggregation} (PRA) algorithm, whose pseudocode is presented in Algorithm~\ref{alg_PRA}.

\begin{algorithm}[t]
\caption{Plug-in Robust Aggregation Algorithm}\label{alg_PRA}
\begin{algorithmic}[1]
\STATE \textbf{Parameters:} Penalty coefficient $\gamma>0$, $m = |E|$
\STATE \textbf{Initialization:} $\hat{\mathbf{P}}_0(e) = 0$ for all $e\in E$
\FOR{$t=1, 2,\ldots,T$}
    \STATE Compute $x_t\in (\Delta(\mathcal{A}))^m$ by solving the saddle point of the problem~\eqref{eq:penalized_SP} using $\hat{\mathbf P}_{t-1}$, with a fixed deterministic tie-breaking rule
    \STATE Observe the current forecast profile $e_t\in E$
    \STATE Sample $a_t\in \mathcal{A}$ based on the strategy $x_t(e_t)$
    \STATE Update the empirical distribution:
    \STATE \begin{equation*}
        \hat{\mathbf{P}}_{t}(e) = \frac{t-1}{t} \hat{\mathbf{P}}_{t-1}(e) + \frac{1}{t}\mathbf{1}_{\{e_t = e\}}
    \end{equation*}
\ENDFOR
\end{algorithmic}
\end{algorithm}

The following theorem states that PRA algorithm achieves sublinear robust regret. 

\begin{theorem}\label{thm_online_regret}
Suppose that $\mathcal C(\mathbf P)$ is nonempty and that the penalty coefficient $\gamma$ is sufficiently large. Then the robust regret of Algorithm~\ref{alg_PRA} satisfies
\[
    \regret
    =
    O\bigl((U+n\gamma)\sqrt{mT}+n\gamma\varepsilon_T T\bigr).
\]
In particular, when $\varepsilon_T=0$, we have $\regret=O(\sqrt T)$, and hence $\regret/T\to0$ as $T\to\infty$.
\end{theorem}

\begin{remark}
    The proof shows that it suffices to choose $\gamma\ge 4UH$, where $H$ is the Hoffman constant associated with the constraint matrix induced by the calibration constraints and the box constraints. In particular, this choice of $\gamma$ is independent of the time horizon $T$ and of the right-hand side $\mathbf P$. The role of this condition is to ensure that the penalized objective $\widetilde W_{\mathbf P}^{\gamma}(x)$ uniformly approximates the constrained robust objective $W_{\mathbf P}(x)$ up to a controlled error.
\end{remark}

\begin{remark}
    The term $n\gamma\varepsilon_T T$ captures the price of approximate calibration. In the online calibration literature, existing forecasting algorithms achieve cumulative calibration error of order $O(\sqrt T)$ in several standard formulations; see, e.g., \citet{gupta2022faster,qiao2024distance,arunachaleswaran2024elementary}. Such cumulative guarantees imply the normalized calibration error used here only under a lower-bound condition on the mass of the relevant reports, or under an unnormalized calibration metric. Under that additional condition, whenever the calibration error in our model satisfies $\varepsilon_T=O(T^{-1/2})$, the additional calibration-error term is $O(n\gamma\sqrt T)$, and the robust regret remains $O(\sqrt T)$.
\end{remark}

We introduce a useful notation for the worst-case expected reward. For the true distribution $\mathbf{P}$ and any decision rule $x\in(\Delta(\mathcal{A}))^m$, define
\begin{equation}\label{eq:WP_def_proof}
W_{\mathbf{P}}(x)
:=
\min_{\rho\in\mathcal{C}(\mathbf{P})}
\sum_{e\in E}\mathbf{P}(e)\,x(e)^\top v_\rho(e),
\end{equation}
and define the robust value
\begin{equation*}
    V(\mathbf{P}) := \max_{x\in(\Delta(\mathcal{A}))^m} W_{\mathbf{P}}(x).
\end{equation*}
By definition, 
\begin{equation*}
    \robust = T\cdot V(\mathbf{P}).
\end{equation*}


Given any distribution $\mathbf{Q}\in \Delta(E)$(including $\mathbf{P}$ and $\hat{\mathbf P}_0$) and any decision rule $x\in (\Delta(\mathcal{A}))^m$, define the \textit{calibration violation} of $\rho\in [0, 1]^m$ under $\mathbf{Q}$ to be 
\begin{equation}\label{eq:violation definition}
    \operatorname{viol}_\mathbf{Q}(\rho) = \sum_{i\in[n]}\sum_{e^i\in E^i}
    \left|
    \sum_{e:\,[e]_i=e^i}\mathbf{Q}(e)\bigl(\rho(e) - e^i\bigr)
    \right|.
\end{equation}
Define 
\begin{equation}\label{eq:penalized_wcv}
    \widetilde{W}_\mathbf{Q}^\gamma(x) := \min_{\rho\in [0, 1]^m} \left\{\sum_{e\in E}\mathbf{Q}(e)x(e)^\top v_\rho(e) + \gamma \operatorname{viol}_{\mathbf{Q}}(\rho)\right\}.
\end{equation}
We also define the corresponding penalized robust value by
\begin{equation}
    \widetilde{V}^\gamma(\mathbf{Q}) := \max_{x\in (\Delta(\mathcal{A}))^m} \widetilde{W}_\mathbf{Q}^\gamma(x).
\end{equation}

We record two immediate properties of these quantities.
\begin{enumerate}
    \item For any $\rho\in \mathcal{C}(\mathbf{P})$, 
    \begin{equation*}
        \mathrm{viol}_{\mathbf{P}}(\rho) \leq \sum_{i\in[n]}\sum_{e^i\in E^i}\sum_{e:\,[e]_i=e^i}\mathbf{P}(e)\varepsilon_T = n\varepsilon_T.
    \end{equation*}
    Hence, for any $x\in (\Delta(\mathcal{A}))^m$,
    \begin{equation}\label{eq:direct_property_1}
        \widetilde{W}_\mathbf{P}^\gamma(x) \leq \min_{\rho\in \mathcal{C}(\mathbf{P})} \Bigl\{\sum_{e\in E}\mathbf{P}(e)x(e)^\top v_\rho(e) + \gamma \operatorname{viol}_{\mathbf{P}}(\rho)\Bigr\} \leq  W_\mathbf{P}(x) + n\gamma\varepsilon_T.
    \end{equation}
    \item For any empirical distribution $\hat{\mathbf{P}}_{t-1}$, let $x_t$ be the decision rule generated by Algorithm~\ref{alg_PRA} at day $t$. By the construction of the algorithm,
    \begin{equation}\label{eq_relation_between_V_W}
        \widetilde{V}^\gamma(\hat{\mathbf{P}}_{t-1}) = \widetilde{W}_{\hat{\mathbf{P}}_{t-1}}^\gamma(x_t).
    \end{equation}
\end{enumerate}

\paragraph{Proof sketch of Theorem~\ref{thm_online_regret}.}
The proof proceeds in several steps. The full proofs of the following lemmas are deferred to Appendix~\ref{ap:proof_of_section_forecast_feecback}. 

\smallskip
\noindent\textbf{(i) Regret decomposition through $V(\cdot), \widetilde{V}^\gamma(\cdot), W_{\cdot}(x)$ and $\widetilde{W}_{\cdot}^\gamma(x)$.} In the PRA algorithm, the decision rule $x_t$ is computed from the empirical distribution $\hat{\mathbf P}_{t-1}$ without using the current forecast profile $e_t$. Consequently, $x_t$ is $\mathcal{F}_{t-1}$-measurable, which allows us to apply Lemma~\ref{lem:iterated_expectation_full} and Eq. \eqref{eq_relation_between_V_W} to express the robust regret as a sum of one-period expected value gaps.

\begin{lemma}\label{lem:regret_decomposition_full}
The robust regret of the PRA algorithm satisfies: 
\begin{equation}\label{eq:regret_decomposition_full}
    \regret
    \;\le\;
    \sum_{t=1}^T
    \mathbb{E}\!\left[V(\mathbf{P})-\widetilde{V}^\gamma(\hat{\mathbf{P}}_{t-1})\right]
    +
    \sum_{t=1}^T
    \mathbb{E}\!\left[\widetilde{W}^\gamma_{\hat{\mathbf{P}}_{t-1}}(x_t)-W_{\mathbf{P}}(x_t)\right].
\end{equation}
\end{lemma}

\smallskip \noindent\textbf{(ii) Bounded difference between $W_\mathbf{P}(x)$ and $\widetilde{W}_{\mathbf{P}}^\gamma(x)$ for sufficiently large $\gamma$.} The regret decomposition in~\eqref{eq:regret_decomposition_full} involves both the constrained value function $W_{\mathbf P}$ and the penalized value function $\widetilde W_{\mathbf P}^{\gamma}$. The key observation is that the calibration constraints define a system of linear inequalities. By Hoffman's lemma, the distance from any $\rho\in[0,1]^m$ to the feasible calibration set $\mathcal C(\mathbf P)$ can be bounded by a constant multiple of its calibration violation, together with the calibration tolerance $\varepsilon_T$. Therefore, any posterior mapping $\rho$ that violates calibration can be projected to a feasible calibrated mapping at a payoff loss controlled by its violation and by $\varepsilon_T$. When $\gamma$ is chosen large enough, the penalty term in $\widetilde W_{\mathbf P}^{\gamma}$ dominates this loss. As a result, minimizing the penalized objective over all $\rho\in[0,1]^m$ yields essentially the same value as minimizing the original payoff over $\rho\in\mathcal C(\mathbf P)$, up to the calibration-error term. Combined with the direct upper bound in~\eqref{eq:direct_property_1}, this gives the following uniform comparison.

\begin{lemma}\label{lem:equiv_penalized_origin}
    If the penalty coefficient $\gamma$ is sufficiently large, then for any $x\in (\Delta(\mathcal{A}))^m$, 
    \begin{equation}\label{eq:equiv_pena_orig}
        \Bigl|\widetilde{W}_\mathbf{P}^\gamma(x) - W_\mathbf{P}(x)\Bigr|\leq n\gamma\varepsilon_T.
    \end{equation}
    Consequently, 
    \begin{equation*}
        \Bigl|\widetilde{V}^\gamma(\mathbf{P}) - V(\mathbf{P})\Bigr|\leq n\gamma\varepsilon_T.
    \end{equation*}
\end{lemma}
Then, the regret bound turns to be
\begin{equation}\label{eq:regret_decomposition}
    \regret
     \le
    \sum_{t=1}^T
    \mathbb{E}\!\left[\widetilde{V}^\gamma(\mathbf{P})-\widetilde{V}^\gamma(\hat{\mathbf{P}}_{t-1})\right]
    +
    \sum_{t=1}^T
    \mathbb{E}\!\left[\widetilde{W}_{\hat{\mathbf{P}}_{t-1}}^{\gamma}(x_t)-\widetilde{W}_{\mathbf{P}}^\gamma(x_t)\right] + 2n\gamma\varepsilon_T\cdot T.
\end{equation}
This regret decomposition reduces the analysis to (i) establishing Lipschitz continuity of $\widetilde{V}^\gamma(\cdot)$ and $\widetilde{W}^\gamma_{\cdot}(x)$, and (ii) controlling the estimation error $\|\hat{\mathbf{P}}_{t-1}-\mathbf{P}\|_1$.

\smallskip
\noindent\textbf{(iii) Lipschitz continuity of $\widetilde{V}^\gamma(\cdot)$ and $\widetilde{W}_{\cdot}^\gamma(x)$.} The Lipschitz bound follows directly from the primal definition of the penalized value functions. For a fixed decision rule $x$ and posterior mapping $\rho$, changing the distribution from $\mathbf Q$ to $\mathbf Q'$ changes the payoff term by at most $U\|\mathbf Q-\mathbf Q'\|_1$. The calibration-violation penalty changes by at most $n\gamma\|\mathbf Q-\mathbf Q'\|_1$, because $|\rho(e)-e^i|\le 1$ and each forecast profile contributes once for each expert. Taking the minimum over $\rho$ and then the maximum over $x$ preserves this Lipschitz bound.

\begin{lemma}\label{lem_lipschitz}
The functions $\widetilde{V}^\gamma(\cdot)$ and $\widetilde{W}^\gamma_{\cdot}(x)$ are Lipschitz-continuous with respect to the underlying distribution with a uniform Lipschitz constant. Specifically, for any distributions $\mathbf{Q}$ and $ \mathbf{Q}^\prime$,
\begin{equation}\label{eq:lipschitz_VW}
    \begin{aligned}
        \sup_{x\in(\Delta(\mathcal{A}))^m}|\widetilde{W}^\gamma_\mathbf{Q}(x) - \widetilde{W}^\gamma_\mathbf{Q^\prime}(x)| &\leq (U+n\gamma)\|\mathbf{Q} - \mathbf{Q^\prime}\|_1, \\ 
        |\widetilde{V}^\gamma(\mathbf{Q}) - \widetilde{V}^\gamma(\mathbf{Q^\prime})|&\leq (U+n\gamma)\|\mathbf{Q} - \mathbf{Q^\prime}\|_1.
    \end{aligned}
\end{equation}
\end{lemma}

\smallskip
\noindent\textbf{(iv) Estimation error of the empirical distribution.}
Since the forecast profiles are sampled independently from the true distribution $\mathbf P$, the empirical frequency $\hat{\mathbf P}_{t-1}(e)$ is an unbiased estimate of $\mathbf P(e)$ for each $e\in E$ and its variance is of order $1/(t-1)$. Summing these coordinate-wise estimation errors over the finite profile space yields the following bound.

\begin{lemma}\label{lem_esti_error}
    The expected error of the empirical distribution $\hat{\mathbf{P}}_{t-1}$ with respect to the true distribution $\mathbf{P}$ satisfies 
    \begin{equation}\label{eq:empirical_error_bound}
    \begin{aligned}
        &\mathbb{E}\bigl[\|\hat{\mathbf{P}}_{t-1}-\mathbf{P}\|_1\bigr] = 1, \qquad\qquad t = 1,\\ &\mathbb{E}\!\left[\|\hat{\mathbf{P}}_{t-1}-\mathbf{P}\|_1\right]
    \le
    \sqrt{\frac{m-1}{t-1}}, \ \  t\geq 2.
    \end{aligned}
\end{equation}
\end{lemma}

Now, we can give a formal proof of theorem~\ref{thm_online_regret}. 

\begin{proof}[Proof of Theorem~\ref{thm_online_regret}]
    By Lemma~\ref{lem:regret_decomposition_full} and ~\ref{lem:equiv_penalized_origin}, we have
    \begin{equation*}
        \regret
         \le
        \sum_{t=1}^T
        \mathbb{E}\!\left[\widetilde{V}^\gamma(\mathbf{P})-\widetilde{V}^\gamma(\hat{\mathbf{P}}_{t-1})\right]
        +
        \sum_{t=1}^T
        \mathbb{E}\!\left[\widetilde{W}_{\hat{\mathbf{P}}_{t-1}}^{\gamma}(x_t)-\widetilde{W}_{\mathbf{P}}^\gamma(x_t)\right] + 2n\gamma \varepsilon_T T.
    \end{equation*}
    Applying Lemma~\ref{lem_lipschitz} and~\ref{lem_esti_error} leads to
    \begin{equation*}
        \begin{aligned}
            \regret
         &\le
        2\sum_{t=1}^T(U+n\gamma)
        \mathbb{E}\!\left[\|\mathbf{P}-\hat{\mathbf{P}}_{t-1}\|_1\right]+2n\gamma \varepsilon_T T\\
        &\leq (2U+2n\gamma)\sum_{t=2}^T\sqrt{\frac{m-1}{t-1}} + (2U+2n\gamma) + 2n\gamma \varepsilon_T T\\
        & \leq (4U+4n\gamma)\sqrt{(m-1)T} + (2U+2n\gamma) + 2n\gamma \varepsilon_T T\\
        & = O\bigl((U+n\gamma)\sqrt{mT} + n\gamma\varepsilon_T T\bigr).
        \end{aligned}
    \end{equation*}
\end{proof}

\subsection{An Example Illustrating the Necessity of Assumption~\ref{as:fixed_distribution}}
\label{subsec:impossibility}

To see why Assumption~\ref{as:fixed_distribution} on the generation of forecast profiles is needed, consider the following binary-state example. Let $\Omega=\{0,1\}$, where $\omega=1$ denotes rain, and let the decision-maker choose an action $a\in\{0,1\}$. The payoff is
\begin{equation*}
u(a,\omega)=\mathbf 1\{a=\omega\}.
\end{equation*}
Thus, if the posterior probability of rain under a forecast profile $e$ is $\rho(e)$, then action $1$ yields expected payoff $\rho(e)$, while action $0$ yields expected payoff $1-\rho(e)$.

There are two experts, each of whom may report one of the values $0,1/2,1$. Consider the following forecast profiles:
\begin{equation*}
e^0=\left(\frac12,\frac12\right),\qquad
e^L_1=\left(\frac12,0\right),\qquad
e^L_2=\left(0,\frac12\right),
\end{equation*}
and
\begin{equation*}
e^H_1=\left(\frac12,1\right),\qquad
e^H_2=\left(1,\frac12\right).
\end{equation*}
Assume for simplicity that $T$ is divisible by $3$. We construct two possible forecast sequences. In the first sequence, the first $T/3$ days have profile $e^0$, the next $T/3$ days have profile $e^L_1$, and the final $T/3$ days have profile $e^L_2$. The resulting empirical distribution assigns probability $1/3$ to each of $e^0,e^L_1,e^L_2$. The calibration constraints then imply
\begin{equation*}
\rho(e^L_1)=0,\qquad \rho(e^L_2)=0.
\end{equation*}
Moreover, conditioning on the days on which Expert 1 reports $1/2$ gives
\begin{equation*}
\frac{\rho(e^0)+\rho(e^L_1)}{2}=\frac12,
\end{equation*}
and hence $\rho(e^0)=1$. The same conclusion is obtained from the calibration constraint for Expert 2's report $1/2$. Therefore, under this empirical distribution, the robust-maxmin benchmark chooses action $1$ on profile $e^0$.

In the second sequence, the first $T/3$ days are again all equal to $e^0$, but the remaining days consist of $T/3$ repetitions of $e^H_1$ and $T/3$ repetitions of $e^H_2$. The resulting empirical distribution assigns probability $1/3$ to each of $e^0,e^H_1,e^H_2$. The calibration constraints imply
\begin{equation*}
\rho(e^H_1)=1,\qquad \rho(e^H_2)=1.
\end{equation*}
Using again the calibration constraint for Expert 1's report $1/2$, we obtain
\begin{equation*}
\frac{\rho(e^0)+\rho(e^H_1)}{2}=\frac12,
\end{equation*}
and hence $\rho(e^0)=0$. The same conclusion follows from Expert 2's report $1/2$. Thus, under this empirical distribution, the robust-maxmin benchmark chooses action $0$ on profile $e^0$.

The two sequences are indistinguishable to the decision-maker during the first $T/3$ days, since the observed forecast profile is $e^0$ on every one of these days in both sequences. Therefore, any online algorithm must use the same mixed action on these days in the two sequences. Let $q_t$ denote the probability with which the algorithm chooses action $1$ on day $t\le T/3$. In the first sequence, the benchmark payoff on profile $e^0$ is achieved by action $1$, so the regret contribution on day $t$ is $1-q_t$. In the second sequence, the benchmark payoff on profile $e^0$ is achieved by action $0$, so the regret contribution on day $t$ is $q_t$. Hence the two regret contributions sum to $1$ on each of the first $T/3$ days. It follows that at least one of the two sequences gives regret at least
\begin{equation*}
\frac12\cdot \frac{T}{3}=\frac{T}{6}.
\end{equation*}
Consequently, no online algorithm can guarantee sublinear robust regret uniformly over arbitrary deterministic orderings of forecast profiles. This illustrates why the assumption that forecast profiles are sampled from a fixed distribution is needed. Without this assumption, the final empirical distribution may induce a robust-maxmin benchmark whose optimal action for an early forecast profile cannot be inferred from the information available when that action is chosen.

\section{Online State Feedback Model}\label{sec:online_state_feedback}

In this section, we study the decision-making problem under state feedback. In this feedback model, the decision-maker observes the forecast profile $e_t$ before choosing an action and observes the realized state $\omega_t$ after the decision is made. Therefore, the realized payoff of every action can be evaluated ex post from the observed state. This additional feedback allows us to benchmark the decision-maker against the best context-dependent decision rule in hindsight, rather than against the robust-maxmin benchmark.

Specifically, we treat forecast profiles as contexts and actions as arms. This reduces the problem to a contextual full-information online learning problem with a finite context space $E$ and a finite action set $\mathcal A$. The natural benchmark is the best mapping $f:E\to\mathcal A$ from forecast profiles to actions in hindsight. This benchmark represents the best fixed context-dependent rule that could have been chosen after observing the entire sequence of forecast profiles and realized states.

\begin{definition}[Contextual Benchmark]
    Given a realized sequence $\{(e_t,\omega_t)\}_{t=1}^T$, the contextual optimal-in-hindsight benchmark is defined as
    \begin{equation}\label{eq_contextual_benchmark}
        \mathrm{Cont}
        :=
        \max_{f:E\to\mathcal{A}}
        \sum_{t=1}^T
        u(f(e_t),\omega_t).
    \end{equation}
\end{definition}

Given an online algorithm that chooses actions $\{a_t\}_{t=1}^T$, its contextual regret is defined as the gap between the contextual benchmark and the cumulative payoff achieved by the algorithm.

\begin{definition}[Contextual Regret]
    The contextual regret of an online algorithm is defined as
    \begin{equation}\label{eq_contextual_regret}
        \mathrm{Reg}^{\prime}
        :=
        \max_{f:E\to\mathcal{A}}
        \mathbb{E}
        \left[
            \sum_{t=1}^T
            \bigl(
                u(f(e_t),\omega_t)
                -
                u(a_t,\omega_t)
            \bigr)
        \right],
    \end{equation}
    where the expectation is taken over the algorithm's internal randomness.
\end{definition}

The benchmark in~\eqref{eq_contextual_benchmark} is natural in the state-feedback model. Since the realized state is observed after each day, the decision-maker can evaluate the payoff that each action would have obtained at that day. Hence, learning can be performed separately for each forecast profile. In particular, for each $e\in E$, the decision-maker maintains an independent online learning algorithm over the action set $\mathcal{A}$, and updates only the algorithm corresponding to the realized forecast profile. The resulting procedure is summarized in Algorithm~\ref{alg_contextual_hedge}.

\begin{algorithm}[t]
\caption{Contextual Hedge Algorithm}\label{alg_contextual_hedge}
\begin{algorithmic}[1]
\STATE \textbf{Parameters:} Learning rate $\eta=\frac{1}{U}\sqrt{\frac{m\log d}{2T}}$, $m=|E|$
\STATE \textbf{Initialization:} For each $e\in E$, set $w_{1,e}(a)=1$ for all $a\in\mathcal{A}$
\FOR{$t=1,2,\ldots,T$}
    \STATE Observe the current forecast profile $e_t\in E$
    \STATE Compute the mixed strategy for $e_t$:
    \STATE \begin{equation*}
        x_t(e_t)(a)
        =
        \frac{w_{t,e_t}(a)}
        {\sum_{a'\in\mathcal{A}} w_{t,e_t}(a')},
        \qquad \forall a\in\mathcal{A}.
    \end{equation*}
    \STATE Sample an action $a_t\in\mathcal{A}$ according to $x_t(e_t)$
    \STATE Observe the realized state $\omega_t$
    \STATE Compute the payoff $u(a,\omega_t)$ for each $a\in\mathcal{A}$
    \STATE Update only the weights associated with $e_t$ by using the Hedge algorithm:
    \STATE \begin{equation*}
        w_{t+1,e_t}(a)
        =
        w_{t,e_t}(a)
        \exp\!\bigl(\eta u(a,\omega_t)\bigr),
        \qquad a\in\mathcal{A}.
    \end{equation*}
    \STATE For all $e\neq e_t$, keep the weights unchanged:
    \STATE \begin{equation*}
        w_{t+1,e}(a)=w_{t,e}(a),
        \quad a\in\mathcal{A}.
    \end{equation*}
\ENDFOR
\end{algorithmic}
\end{algorithm}

The following result shows that this simple algorithm yields no contextual regret.

\begin{theorem}\label{thm_contextual_regret_full_feedback}
    With state feedback and the learning rate in Algorithm~\ref{alg_contextual_hedge}, Algorithm~\ref{alg_contextual_hedge} guarantees
    \begin{equation*}
        \mathrm{Reg}^\prime
        \le
        2U\sqrt{2m\log d\cdot T},
    \end{equation*}
    where $m = |E|$ and $d = |\mathcal{A}|$.
\end{theorem}
\begin{proof}
    See Appendix~\ref{ap:proof_of_online_state_feedback}.
\end{proof}

\begin{remark}
    If the decision-maker observes only the payoff of the sampled action, rather than the realized state, the same reduction still applies by running an independent EXP3 algorithm \citep{auer2002nonstochastic} for each forecast profile. The resulting algorithm achieves no contextual regret under bandit feedback, with regret bound of order $O\!\left(U\sqrt{T|E||\mathcal A|\log|\mathcal A|}\right),$ compared with the bound $O\!\left(U\sqrt{T|E|\log|\mathcal A|}\right)$ obtained under state feedback.
\end{remark}

Algorithm~\ref{alg_contextual_hedge} provides a simple way to achieve no contextual regret, but it is not necessarily optimal in terms of regret bounds. An interesting open question is whether the calibration constraints can be exploited to obtain sharper regret guarantees, and what the optimal regret rate is in the state-feedback model.

\section{Conclusion and Future Work}

This paper develops a framework for decision-making with multiple calibrated forecasts when the decision-maker knows only that experts satisfy calibration. We first introduce a robust max-min benchmark, which is the best payoff that can be guaranteed against all posterior mappings consistent with calibration. This benchmark captures the decision-relevant information contained in the joint pattern of forecasts without requiring knowledge of experts' signal structures, priors, or conditional independence relationships. We show that the robust max-min benchmark is tractable and admits a linear-programming formulation. It also dominates the optimal-in-hindsight benchmark up to calibration error, reflecting the value of conditioning on the full forecast profile rather than using only a single expert. At the same time, it can be strictly below the Bayesian benchmark, highlighting the additional value of knowing the full data-generating process.

We then study online learning under two feedback structures. Under forecast-only feedback, where the decision-maker observes forecast profiles but not realized states or payoffs, we propose a plug-in robust aggregation algorithm and prove sublinear robust regret relative to the robust benchmark when forecast profiles are drawn from a fixed distribution. Under state feedback, where the realized state is observed after each decision, the learner can compete with the stronger contextual optimal-in-hindsight benchmark, and a simple contextual full-information learning algorithm achieves no contextual regret.

Several questions remain open. First, the forecast-only result relies on the assumption that forecast profiles are sampled from a fixed distribution. An important direction is to relax this assumption, for example by allowing adversarially controlled noise around an underlying distribution or time-varying forecast distributions subject to suitable stability. Second, the regret bounds obtained here may not be optimal. It remains open whether sharper rates can be achieved, either in the dependence on \(T\) or in the dependence on \(m\) and \(d\). In the state-feedback model, our algorithm uses a simple reduction to contextual online learning and does not exploit the structure imposed by calibration constraints. Designing algorithms that use these constraints more directly may lead to tighter regret bounds.

\bibliographystyle{apalike}
\bibliography{ref}
\newpage
\appendix

\section{Strong Duality Theorem and Hoffman Lemma}\label{app:duality}

First, we give the standard strong duality theorem for linear programming used to derive the dual representation; see, e.g.,
\cite[Ch.~5]{boyd2004convex} or \cite[Ch.~4]{vanderbei1998linear}.

\begin{theorem}[Strong Duality Theorem]\label{thm:strong_duality_lp}
Consider the primal--dual pair of linear programs
\begin{align*}
\text{(P)}\quad &\min_{z\in\mathbb{R}^n}\ \ c^\top z \quad
\text{s.t.}\quad Gz = h,\ \ z\ge 0,\\
\text{(D)}\quad &\max_{y\in\mathbb{R}^p}\ \ h^\top y \quad
\text{s.t.}\quad G^\top y \le c.
\end{align*}
If \textup{(P)} is feasible and has a finite optimal value (equivalently, admits an optimal solution),
then \textup{(D)} is feasible, attains its optimum, and the optimal values coincide:
\[
\min\textup{(P)}=\max\textup{(D)}.
\]
The symmetric statement holds with the roles of \textup{(P)} and \textup{(D)} interchanged.
\end{theorem}

Then, we present a standard form of Hoffman's lemma for polyhedral systems; see
\cite{hoffman1952approximate}. We use this lemma to relate the $\ell_1$-residual of linear constraints to the distance to the corresponding feasible set.

\begin{theorem}[Hoffman Lemma]\label{thm:hoffman}
Let $M\in\mathbb{R}^{r\times m}$ and $h\in\mathbb{R}^r$, and define the polyhedron
\[
\mathcal{S}:=\{z\in\mathbb{R}^m: Mz\le h\}.
\]
Assume that $\mathcal S$ is nonempty. Then there exists a constant $H>0$ (depending only on $M$) such that for every $z\in\mathbb{R}^m$,
\begin{equation}\label{eq:hoffman_bound_general}
\mathrm{dist}_1(z,\mathcal{S})
\;\le\;
H\,\|(Mz-h)_+\|_1,
\end{equation}
where $(\cdot)_+$ is taken element-wise and $\mathrm{dist}_1(z,\mathcal{S})
:=\inf_{s\in\mathcal{S}}\|z-s\|_1$ denotes the $\ell_1$-distance. Moreover, the same statement holds for equality constraints by writing $Mz=h$ as
$Mz\le h$ and $-Mz\le -h$.
\end{theorem}

\section{Missing Proofs}

\subsection{Proofs of Section~\ref{sec:offline}}\label{ap:proofs_of_section_offline}

\begin{proof}[Proof of Proposition~\ref{thm_indepen_strate_optimal}]
    The two strategy spaces contain different objects, so we compare them through their induced marginals.
    For any correlated distribution $y\in\Delta(\mathcal A^m)$, define the marginal mixed action
    \[
        x_y(e)(a):=\sum_{\alpha\in\mathcal A^m:\alpha(e)=a}y(\alpha),
        \qquad e\in E,\ a\in\mathcal A.
    \]
    Then $x_y\in(\Delta(\mathcal A))^m$, and for every $\rho\in\mathcal C(\edp_T)$,
    \[
        \sum_{\alpha\in\mathcal A^m}y(\alpha)\sum_{e\in E}
        \edp_T(e)\mathbf 1_{\alpha(e)}^\top v_\rho(e)
        =
        \sum_{e\in E}\edp_T(e)x_y(e)^\top v_\rho(e).
    \]
    Thus every correlated strategy has a profile-wise mixed strategy with exactly the same payoff against every feasible $\rho$. Conversely, given any $x\in(\Delta(\mathcal A))^m$, let $y_x$ be the product distribution over pure mappings $\alpha\in\mathcal A^m$ whose coordinate marginal at profile $e$ is $x(e)$. The same identity shows that $y_x$ and $x$ have the same payoff against every $\rho$. Taking the minimum over $\rho$ and then the maximum over strategies on the two sides proves the result.
\end{proof}

\begin{proof}[Proof of Theorem~\ref{prop_compute_RMM}]
    For notational convenience, let $p:=\edp_T$. For any $x\in(\Delta(\mathcal A))^m$, define $d_x,c_x:E\to\mathbb R$ by
    \[
        d_x(e):=x(e)^\top (u_1-u_0),
        \quad
        c_x(e):=x(e)^\top u_0,
        \quad \forall e\in E.
    \]
    where $u_1:=(u(a, 1))_{a\in \mathcal{A}}$ and $u_0:=(u(a, 0))_{a\in \mathcal{A}}$. For any  $\rho:E\to[0,1]$, define $q_\rho:E\to\mathbb R$ by
    \[
        q_\rho(e):=p(e)\rho(e),
        \quad \forall e\in E.
    \]
    Then the objective can be written as
    \[
        \sum_{e\in E}
        p(e)x(e)^\top
        \bigl(\rho(e)u_1+(1-\rho(e))u_0\bigr)
        =
        \sum_{e\in E} \bigl[p(e)c_x(e) + q_\rho(e)d_x(e)\bigr]
        =
        c_x^\top p + d_x^\top q_\rho,
    \]
    where, since $E$ is finite, we identify mappings (i.e., $p, c_x, q_\rho, d_x$) from $E$ to $\mathbb R$ with vectors in $\mathbb R^m$ under an arbitrary fixed ordering of $E$. 
    
    We next rewrite the calibration constraints in terms of $q_\rho$. Let
    \[
        \mathcal J:=\{(i,\zeta): i\in[n],\ \zeta\in E^i\}, 
    \]
    so that $r=|\mathcal J|$. We index the rows by elements of $\mathcal J$ and the columns by elements of $E$. Define the incidence matrix $A\in\{0,1\}^{r\times m}$ by
    \[
        A_{(i,\zeta),e}
        =
        \indicate{[e]_i=\zeta},
        \quad \forall
        (i,\zeta)\in\mathcal J,\ e\in E.
    \]
    Define $b,\epsilon\in\mathbb R^r$ by
    \[
        b_{(i,\zeta)}
        :=
        \left(\sum_{e:[e]_i=\zeta}p(e)\right)\zeta,
        \quad
        \epsilon_{(i,\zeta)}
        :=
        \left(\sum_{e:[e]_i=\zeta}p(e)\right)\varepsilon_T .
    \]
    Then the calibration constraints are equivalently written as
    \[
        Aq_\rho\leq b + \epsilon,
        \quad
        -Aq_\rho\leq \epsilon-b,
        \quad
        0\leq q_\rho\leq p .
    \]
    Denote the corresponding feasible set by
    \[
        \mathcal{F}_p
        :=
        \left\{
        q_\rho\in\mathbb R^m:
        Aq_\rho\leq b+\epsilon,\
        -Aq_\rho\leq \epsilon-b,\
        0\leq q_\rho\leq p
        \right\}.
    \]
    
    Therefore, for a fixed decision rule $x$, the worst-case expected payoff equals
    \[
        c_x^\top p+\min_{q_\rho\in\mathcal{F}_p} d_x^\top q_\rho.
    \]
    For each fixed $x$, this is a linear program in $q_\rho$. Since $\mathcal{C}(\edp_T)$ is nonempty, $\mathcal{F}_p$ is nonempty and the problem is feasible. Moreover, the problem attains a finite optimal value since $\mathcal F_p$ is a compact polytope. Hence, the strong duality theorem~\ref{thm:strong_duality_lp} applies. 

    To derive the dual, introduce slack variables $\xi^+, \xi^-\in \mathbb{R}_+^r, s\in \mathbb R^m_+$ and write the inner minimization as 
    \begin{align*}
        \min_{q_\rho\in\mathbb R^m}
        \quad & \  d_x^\top q_\rho  \\
        \text{s.t.}\quad
        &
        \ Aq_\rho + \xi^+ = b+\epsilon,\\
        & -Aq_\rho + \xi^- = \epsilon-b,\\
        &\ \  q_\rho + s= p, \\
        & q_\rho, \xi^+, \xi^-, s\geq 0. 
    \end{align*}
    Introducing dual variables $\lambda^+, \lambda^-\in\mathbb R^r$ and $\mu \in\mathbb R^m,$ the dual format can be written as
    \begin{align*}
        \max_{\lambda^+,\lambda^-,\mu}
        \quad
        &
        (b+\epsilon)^\top\lambda^+ + (\epsilon-b)^\top\lambda^- + p^\top\mu
        \\
        \text{s.t.}\quad
        &
        A^\top\lambda^+
        -
        A^\top\lambda^-
        +
        \mu
        \leq 
        d_x,
        \\
        &
        \lambda^+,\lambda^-,\mu \leq 0 .
    \end{align*}
    
    Therefore, the robust-maxmin benchmark is equal to $T$ times the value of the linear program: 
    \begin{align}
        \max_{x,\lambda^+,\lambda^-,\mu}
        \quad
        &
        p^\top c_x
        +
        (b+\epsilon)^\top\lambda^+
        +
        (\epsilon-b)^\top\lambda^-
        +
        p^\top\mu
        \label{eq_RMM_LP_dualized}
        \\
        \text{s.t.}\quad
        &
        A^\top\lambda^+
        -
        A^\top\lambda^-
        +
        \mu
        \leq 
        d_x,
        \notag\\
        &
        \lambda^+,\lambda^-, \mu\leq 0,
        \notag\\
        &x(e) \in \Delta(\mathcal{A}),
        \quad \forall e\in E.
        \notag
    \end{align}
    Thus, the linear program contains $md+m+2r$ variables and $md+3m+2r$ constraints. It follows that the robust-maxmin benchmark is computable in time polynomial in $m$, $d$, and $r$ using standard linear-programming algorithms.
\end{proof}

\begin{proof}[Proof of Proposition~\ref{prop:robust_vs_oih}]
    Fix an expert $i\in[n]$, and let $\sigma^\ast:E^i\to\mathcal A$ be a policy that attains $U^i$. Define a decision rule $x^\ast\in(\Delta(\mathcal A))^m$ by
    \[
        x^\ast(e)=\mathbf 1_{\sigma^\ast([e]_i)},
        \quad e\in E,
    \]
    where $\mathbf 1_a$ denotes the unit vector corresponding to action $a$. That is, under forecast profile $e$, the decision-maker chooses the action $\sigma^\ast([e]_i)$. By the definition of the robust-maxmin benchmark,
    \begin{equation*}
        \begin{aligned}
            \robust &\geq \min_{\rho\in \mathcal{C}(\edp_T)} T \sum_{e\in E} \edp_T(e)x^\ast(e)^\top v_\rho(e)\\
            & = \min_{\rho\in \mathcal{C}(\edp_T)} T \sum_{e^i\in E^i}\sum_{e:[e]_i = e^i} \edp_T(e) \Bigl(\rho(e)u(\sigma^\ast(e^i), 1) + (1-\rho(e))u(\sigma^\ast(e^i), 0)\Bigr).
        \end{aligned}
    \end{equation*}
    
    For each $\rho\in\mathcal C(\edp_T)$, define $\kappa^\rho:E^i\to[0,1]$ as follows. If $\edp_T(e^i)>0$, let
    \[
        \kappa^\rho(e^i)
        :=
        \frac{
            \sum_{e:[e]_i=e^i}\edp_T(e)\rho(e)
        }{
            \edp_T(e^i)
        };
    \]
    otherwise, set $\kappa^\rho(e^i)=0$. By the calibration constraint, whenever $\edp_T(e^i)>0$,
    \[
        \left|\kappa^\rho(e^i)-e^i\right|\leq \varepsilon_T.
    \]
    Then, 
    \begin{equation*}
        \robust \geq \min_{\rho\in \mathcal{C}(\edp_T)} T\ \sum_{e^i\in E^i}\edp_T(e^i)
        \Bigl(\kappa^\rho(e^i)u(\sigma^\ast(e^i), 1) + (1-\kappa^\rho(e^i))u(\sigma^\ast(e^i), 0)\Bigr).
    \end{equation*}
    
    Next, for the realized state sequence, define $\kappa:E^i\to[0,1]$ as follows. If $\edp_T(e^i)>0$, let
    \[
        \kappa(e^i)
        :=
        \frac{
            \sum_{t\leq T}\indicate{e_t^i=e^i}\indicate{\omega_t=1}
        }{
            \sum_{t\leq T}\indicate{e_t^i=e^i}
        };
    \]
    otherwise, set $\kappa(e^i)=0$. The calibration constraint~\eqref{eq_calib_deter} leads to
    \[
        \left|\kappa(e^i)-e^i\right|\leq \varepsilon_T.
    \]
    Moreover, by the definition of $\sigma^\ast$,
    \[
        U^i
        =
        T\sum_{e^i\in E^i}
        \edp_T(e^i)
        \Bigl(
            \kappa(e^i)u(\sigma^\ast(e^i),1)
            +
            \bigl(1-\kappa(e^i)\bigr)u(\sigma^\ast(e^i),0)
        \Bigr).
    \]
    
    Using the preceding lower bound on $\robust$ and the above representation of $U^i$, it follows that
    \begin{align*}
        \robust-U^i
        &\geq
        \min_{\rho\in\mathcal C(\edp_T)}
        T\sum_{e^i\in E^i}
        \edp_T(e^i)
        \bigl(\kappa^\rho(e^i)-\kappa(e^i)\bigr)
        \bigl(
            u(\sigma^\ast(e^i),1)-u(\sigma^\ast(e^i),0)
        \bigr)\\
        &\geq
        -T
        \sum_{e^i\in E^i}
        \edp_T(e^i)
        \left|\kappa^\rho(e^i)-\kappa(e^i)\right|
        \left|
            u(\sigma^\ast(e^i),1)-u(\sigma^\ast(e^i),0)
        \right|\\
        &\geq
        -4\varepsilon_T U T.
    \end{align*}
    where the last inequality holds because $\vert\kappa^\rho(e^i) - \kappa(e^i)\vert\leq 2\varepsilon_T$ and $\|u_1 - u_0\|_\infty\leq 2U$. Since the above inequality holds for every expert $i\in[n]$, taking the maximum over $i$ yields
    \begin{equation*}
        \robust \geq \oih -4\varepsilon_TU T.
    \end{equation*}
\end{proof}

\subsection{Proof of Section~\ref{sec:online_forecast_feedback}}\label{ap:proof_of_section_forecast_feecback}

\begin{proof}[Proof of Lemma~\ref{lem:iterated_expectation_full}]
    Fix any $\rho\in\mathcal C(\mathbf P)$. By the tower property of conditional expectation,
    \[
    \mathbb E\!\left[
        \sum_{t=1}^T x_t(e_t)^\top v_\rho(e_t)
    \right]
    =
    \mathbb E\!\left[
        \mathbb E\!\left[
            \sum_{t=1}^T x_t(e_t)^\top v_\rho(e_t)
            \mid \mathcal F_{T-1}
        \right]
    \right].
    \]
    For each $t\leq T-1$, the random variable $x_t(e_t)^\top v_\rho(e_t)$ is $\mathcal F_{T-1}$-measurable. Moreover, $x_T$ is $\mathcal F_{T-1}$-measurable, while $e_T$ is independent of $\mathcal F_{T-1}$. Hence,
    \[
    \mathbb E\!\left[
        \sum_{t=1}^T x_t(e_t)^\top v_\rho(e_t)
        \mid \mathcal F_{T-1}
    \right]
    =
    \sum_{t=1}^{T-1} x_t(e_t)^\top v_\rho(e_t)
    +
    \sum_{e\in E}\mathbf P(e)\,x_T(e)^\top v_\rho(e).
    \]
    Taking expectations on both sides gives
    \begin{align*}
        \mathbb E\!\left[
            \sum_{t=1}^T x_t(e_t)^\top v_\rho(e_t)
        \right]
        =
        \mathbb E\!\left[
            \sum_{t=1}^{T-1} x_t(e_t)^\top v_\rho(e_t)
        \right] +
        \mathbb E\!\left[
            \sum_{e\in E}\mathbf P(e)\,x_T(e)^\top v_\rho(e)
        \right].
    \end{align*}
    Repeating the same argument with respect to $\mathcal F_{T-2},\ldots,\mathcal F_0$, we obtain
    \[
        \mathbb E\!\left[
            \sum_{t=1}^T x_t(e_t)^\top v_\rho(e_t)
        \right]
        =
        \sum_{t=1}^T
        \mathbb E\!\left[
            \sum_{e\in E}\mathbf P(e)\,x_t(e)^\top v_\rho(e)
        \right].
    \]
    The representation in~\eqref{eq:regret_equiv_form_full} then follows from the definition of robust regret and the robust-maxmin benchmark.
\end{proof}

\begin{proof}[Proof of Lemma~\ref{lem:regret_decomposition_full}]
    Since $x_t$ generated by the PRA algorithm is $\mathcal{F}_{t-1}$-measurable, by Lemma~\ref{lem:iterated_expectation_full}:
    \begin{equation*}
        \regret
    =
    T\cdot V(\mathbf{P})
    -
    \min_{\rho\in \mathcal{C}(\mathbf{P})}
    \sum_{t=1}^T
    \mathbb{E}\!\left[
    \sum_{e\in E}\mathbf{P}(e)\,x_t(e)^\top v_\rho(e)
    \right].
    \end{equation*}
    For any $t$ and any $\rho\in\mathcal{C}(\mathbf{P})$,
    \[
    \sum_{e\in E}\mathbf{P}(e)\,x_t(e)^\top v_\rho(e)
    \;\ge\;
    \min_{\rho'\in\mathcal{C}(\mathbf{P})}
    \sum_{e\in E}\mathbf{P}(e)\,x_t(e)^\top v_{\rho'}(e)
    =
    W_{\mathbf{P}}(x_t).
    \]
    Hence
    \[
    \min_{\rho\in\mathcal{C}(\mathbf{P})}
    \sum_{t=1}^T\mathbb{E}\!\left[\sum_{e\in E}\mathbf{P}(e)\,x_t(e)^\top v_\rho(e)\right]
    \;\ge\;
    \sum_{t=1}^T \mathbb{E}\!\Bigl[W_{\mathbf{P}}(x_t)\Bigr].
    \]
    Plugging this into the robust regret gives
    \[
    \regret
    \le
    \sum_{t=1}^T \mathbb{E}\!\Bigl[V(\mathbf{P})-W_{\mathbf{P}}(x_t)\Bigr].
    \]
    Now add and subtract $\widetilde{V}^{\gamma}(\hat{\mathbf{P}}_{t-1})$ and use $\widetilde{W}^\gamma_{\hat{\mathbf{P}}_{t-1}}(x_t)=\widetilde{V}^\gamma(\hat{\mathbf{P}}_{t-1})$:
    \begin{equation*}
        \begin{aligned}
            \regret
    &\le \sum_{t=1}^T \mathbb{E}\!\Bigl[\bigl(V(\mathbf{P})-\widetilde{V}^\gamma(\hat{\mathbf{P}}_{t-1})\bigr)+\bigl(\widetilde{W}^\gamma_{\hat{\mathbf{P}}_{t-1}}(x_t)-W_{\mathbf{P}}(x_t)\bigr)\Bigr]\\
            &= \sum_{t=1}^T
    \mathbb{E}\!\left[V(\mathbf{P})-\widetilde{V}^\gamma(\hat{\mathbf{P}}_{t-1})\right]
    +
    \sum_{t=1}^T
    \mathbb{E}\!\left[\widetilde{W}^\gamma_{\hat{\mathbf{P}}_{t-1}}(x_t)-W_{\mathbf{P}}(x_t)\right].
        \end{aligned}
    \end{equation*}
\end{proof}

\begin{proof}[Proof of Lemma~\ref{lem:equiv_penalized_origin}]
    First, we introduce some notations to give cleaner expressions of $W_{\mathbf{P}}(x)$ and $\widetilde{W}^\gamma_{\mathbf{P}}(x)$. This expression of $\widetilde{W}^\gamma_{\mathbf{P}}(x)$ will also be useful for establishing its Lipschitz continuity. 
    
    
    Fix a decision rule $x\in (\Delta(\mathcal{A}))^m$. Define $d,c\in  \mathbb{R}^m$ by
    \begin{equation*}
        d(e) := x(e)^\top (u_1 - u_0),\qquad c(e) := x(e)^\top u_0,
    \end{equation*}
    where $u_1=(u(a, 1))_{a\in \mathcal{A}}$ and $u_0=(u(a, 0))_{a\in \mathcal{A}}$. Since $|u|\leq U$, it follows that
    \begin{equation*}
        \|d\|_\infty \leq 2U,\quad \|c\|_\infty\leq U.
    \end{equation*}
    For any $\rho: E\to [0, 1]$, define $q\in \mathbb{R}^m$ by 
    \begin{equation*}
        q(e):=\mathbf{P}(e)\rho(e),\quad \forall e\in E.
    \end{equation*}
    Then, the objective can be written as
    \begin{equation*}
        \sum_{e\in E}\mathbf{P}(e)\, x(e)^\top \bigl(\rho(e)u_1 + (1-\rho(e))u_0\bigr)
        =
        \mathbf{P}^\top c  + d^\top q.
    \end{equation*}


    Introduce the index set $\mathcal{J}$ and the incidence matrix $A$ defined in the proof of Theorem~\ref{prop_compute_RMM}. Define the right-hand side vector $b(\mathbf{P})\in\mathbb{R}^r$ by
    \begin{equation*}
        b_{(i,\zeta)}(\mathbf{P})
        :=
        \Bigl(\sum_{e:[e]_i=\zeta}\mathbf{P}(e)\Bigr)\cdot \zeta.
    \end{equation*}
    Define the calibration error vector $\epsilon(\mathbf{P})\in\mathbb{R}^r$ by
    \begin{equation*}
        \epsilon_{(i,\zeta)}(\mathbf{P})
        :=
        \Bigl(\sum_{e:[e]_i=\zeta}\mathbf{P}(e)\Bigr)\cdot \varepsilon_T.
    \end{equation*}
    Then, the calibration constraint becomes
    \begin{gather*}
        Aq \leq b(\mathbf{P}) + \epsilon(\mathbf{P}),\quad  -Aq \leq \epsilon(\mathbf{P}) - b(\mathbf{P}),\quad
        \mathbf0\leq q\leq \mathbf{P}.
    \end{gather*}
    Denote the feasible set of such $q\in \mathbb{R}^m$ by $\mathcal{F}(\mathbf{P})$, which is a polyhedron. Therefore, 
    \begin{equation}\label{eq_lp_format}
        W_\mathbf{P}(x) = \mathbf{P}^\top c + \min_{q\in\mathcal{F}(\mathbf{P})} d^\top q.
    \end{equation}

    For any $\rho\in [0, 1]^m$, the calibration violation $\operatorname{viol}_\mathbf{P}(\rho)$ equals $\|Aq - b(\mathbf{P})\|_1$. Hence, 
    \begin{equation}\label{eq:simple_format_penal}
        \widetilde{W}_\mathbf{P}^\gamma(x) = \mathbf{P}^\top c + \min_{0\leq q\le \mathbf{P}} \bigl\{d^\top q + \gamma \|Aq-b(\mathbf{P})\|_1\bigr\}.
    \end{equation}

    Then, we prove the results by applying the Hoffman lemma to bound the distance between $q$ and $\mathcal{F}(\mathbf{P})$ with the calibration violation $\operatorname{viol}_\mathbf{P}(\rho)$.

    By the Hoffman lemma, there exists a constant $H>0$ depending only on the constraint matrix such that, whenever $\mathcal F(\mathbf P)$ is nonempty, for any $q\in \mathbb{R}^m$ with $\mathbf0\leq q\leq \mathbf{P}$, there exists a $\bar{q}\in \mathcal{F}(\mathbf{P})$ satisfying that
    \begin{equation*}
        \begin{aligned}
            \|q-\bar{q}\|_1&\leq H\bigl(\|(Aq - b(\mathbf{P}) - \epsilon(\mathbf{P}))_+\|_1 + (-Aq + b(\mathbf{P}) - \epsilon(\mathbf{P}))_+\|_1\bigr)\\
            &\leq 2H \bigl( \|Aq - b(\mathbf{P})\|_1 + \|\epsilon(\mathbf{P})\|_1 \bigr)\\
            &= 2H \bigl( \|Aq - b(\mathbf{P})\|_1 + n\varepsilon_T \bigr).
        \end{aligned}
    \end{equation*}
    Then, for any $q\in \mathbb{R}^m$ with $\mathbf0\leq q\leq \mathbf{P}$:
    \begin{equation*}
        d^\top \bar{q}\leq d^\top q + \|d\|_\infty \|q-\bar{q}\|_1\leq d^\top q + 4UH\bigl(\|Aq-b(\mathbf{P})\|_1 + n\varepsilon_T\bigr).
    \end{equation*}
    Since $\bar{q}\in \mathcal{F}(\mathbf{P})$, by the definition~\eqref{eq_lp_format} of $W_\mathbf{P}(x)$, 
    \begin{equation*}
        \begin{aligned}
            W_\mathbf{P}(x)\leq \mathbf{P}^\top c + d^\top \bar{q} \leq \mathbf{P}^\top c + d^\top q + 4UH\bigl(\|Aq-b(\mathbf{P})\|_1 + n\varepsilon_T\bigr).
        \end{aligned}
    \end{equation*}
    Since $\gamma\ge 4UH$, we further have
    \begin{equation*}
        W_\mathbf{P}(x) \le \mathbf{P}^\top c + d^\top q + \gamma \bigl(\|Aq-b(\mathbf{P})\|_1 + n\varepsilon_T\bigr).
    \end{equation*}
    Taking minimum over all $q\in \mathbb{R}^m$ with $ 0\leq q\leq \mathbf{P}$ leads to
    \begin{equation*}
        W_\mathbf{P}(x) \leq \widetilde{W}_\mathbf{P}^\gamma (x) + n\gamma \varepsilon_T.
    \end{equation*}
    Combined with the direct property~\eqref{eq:direct_property_1} that $ \widetilde{W}_\mathbf{P}^\gamma (x)\le W_\mathbf{P}(x) + n\gamma \varepsilon_T$, we prove that 
    \begin{equation*}
        \bigl|\widetilde{W}_\mathbf{P}^\gamma (x)- W_\mathbf{P}(x)\bigr|\leq n\gamma \varepsilon_T,\qquad \forall x\in (\Delta(\mathcal{A}))^m.
        \end{equation*}
        Taking maximum over all $x\in (\Delta(\mathcal{A}))^m$ and using $\bigl|\max_x f(x)-\max_x g(x)\bigr|\leq \sup_x |f(x)-g(x)|$ leads to
        \begin{equation*}
            \bigl|\widetilde{V}^\gamma(\mathbf{P}) - V(\mathbf{P})\bigr|\leq n\gamma\varepsilon_T.
        \end{equation*}
\end{proof}

\begin{proof}[Proof of Lemma~\ref{lem_lipschitz}]
    Fix any decision rule $x\in(\Delta(\mathcal A))^m$. For a distribution $\mathbf Q\in\Delta(E)$ and a posterior mapping $\rho\in[0,1]^m$, define
    \[
        \Phi_{\mathbf Q}(x,\rho)
        :=
        \sum_{e\in E}\mathbf Q(e)x(e)^\top v_\rho(e)
        +
        \gamma\operatorname{viol}_{\mathbf Q}(\rho).
    \]
    Then
    \[
        \widetilde W_{\mathbf Q}^{\gamma}(x)
        =
        \min_{\rho\in[0,1]^m}\Phi_{\mathbf Q}(x,\rho).
    \]
    We first bound the pointwise change in $\Phi_{\mathbf Q}(x,\rho)$ when $\mathbf Q$ is replaced by $\mathbf Q'$. Since $x(e)^\top v_\rho(e)$ is a convex combination of payoffs and $|u|\le U$, we have
    \[
        \left|x(e)^\top v_\rho(e)\right|\le U,
        \qquad \forall e\in E.
    \]
    Hence,
    \[
        \left|
        \sum_{e\in E}\bigl(\mathbf Q(e)-\mathbf Q'(e)\bigr)x(e)^\top v_\rho(e)
        \right|
        \le
        U\|\mathbf Q-\mathbf Q'\|_1.
    \]
    Next, by the definition of $\operatorname{viol}_{\mathbf Q}$ and the triangle inequality,
    \begin{align*}
        &\left|
        \operatorname{viol}_{\mathbf Q}(\rho)
        -
        \operatorname{viol}_{\mathbf Q'}(\rho)
        \right| \\
        &\quad\le
        \sum_{i\in[n]}\sum_{e^i\in E^i}
        \left|
        \sum_{e:\,[e]_i=e^i}
        \bigl(\mathbf Q(e)-\mathbf Q'(e)\bigr)
        \bigl(\rho(e)-e^i\bigr)
        \right| \\
        &\quad\le
        \sum_{i\in[n]}\sum_{e^i\in E^i}
        \sum_{e:\,[e]_i=e^i}
        \left|\mathbf Q(e)-\mathbf Q'(e)\right|
        \left|\rho(e)-e^i\right| \\
        &\quad\le
        \sum_{i\in[n]}\sum_{e\in E}
        \left|\mathbf Q(e)-\mathbf Q'(e)\right|
        =
        n\|\mathbf Q-\mathbf Q'\|_1,
    \end{align*}
    where the last inequality uses $\rho(e),e^i\in[0,1]$ and the fact that, for each fixed expert $i$, every profile $e$ contributes to exactly one forecast value $e^i=[e]_i$. Therefore, uniformly over $x$ and $\rho$,
    \begin{equation}\label{eq:pointwise_lipschitz_penalized_objective}
        \left|
        \Phi_{\mathbf Q}(x,\rho)-\Phi_{\mathbf Q'}(x,\rho)
        \right|
        \le
        (U+n\gamma)\|\mathbf Q-\mathbf Q'\|_1.
    \end{equation}

    We now pass this pointwise bound to the minimized value. For any two functions $f$ and $g$ on the same domain,
    \[
        \left|\min_z f(z)-\min_z g(z)\right|
        \le
        \sup_z |f(z)-g(z)|.
    \]
    Applying this inequality to $\Phi_{\mathbf Q}(x,\cdot)$ and $\Phi_{\mathbf Q'}(x,\cdot)$, together with~\eqref{eq:pointwise_lipschitz_penalized_objective}, gives
    \[
        \left|
        \widetilde W_{\mathbf Q}^{\gamma}(x)
        -
        \widetilde W_{\mathbf Q'}^{\gamma}(x)
        \right|
        \le
        (U+n\gamma)\|\mathbf Q-\mathbf Q'\|_1.
    \]
    Since the bound is independent of $x$, we obtain
    \[
        \sup_{x\in(\Delta(\mathcal A))^m}
        \left|
        \widetilde W_{\mathbf Q}^{\gamma}(x)
        -
        \widetilde W_{\mathbf Q'}^{\gamma}(x)
        \right|
        \le
        (U+n\gamma)\|\mathbf Q-\mathbf Q'\|_1.
    \]

    Finally, using
    \[
        \left|\max_x f(x)-\max_x g(x)\right|
        \le
        \sup_x |f(x)-g(x)|,
    \]
    we get
    \[
        \left|
        \widetilde V^\gamma(\mathbf Q)-\widetilde V^\gamma(\mathbf Q')
        \right|
        \le
        \sup_{x\in(\Delta(\mathcal A))^m}
        \left|
        \widetilde W_{\mathbf Q}^{\gamma}(x)
        -
        \widetilde W_{\mathbf Q'}^{\gamma}(x)
        \right|
        \le
        (U+n\gamma)\|\mathbf Q-\mathbf Q'\|_1.
    \]
\end{proof}

\begin{proof}[Proof of Lemma~\ref{lem_esti_error}]
For $t=1$, since $\hat{\mathbf{P}}_0(e) = 0$ for all $e\in E$, 
\begin{equation*}
    \mathbb{E}\bigl[\|\hat{\mathbf{P}}_{t-1}-\mathbf{P}\|_1\bigr]  = \mathbb{E}\bigl[\|\mathbf{P}\|_1\bigr] = 1.
\end{equation*}

For $t\geq 2$, define the Pearson $\chi^2$-divergence
\[
\chi^2(\hat{\mathbf{P}}_{t-1},\mathbf{P})
:=
\sum_{e:\mathbf{P}(e)>0}\frac{(\hat{\mathbf{P}}_{t-1}(e)-\mathbf{P}(e))^2}{\mathbf{P}(e)}.
\]
Since $\hat{\mathbf{P}}_{t-1}(e)$ is the empirical mean of a Bernoulli distribution with the mean $\mathbf{P}(e)$ and $\{e_\tau\}_{\tau\leq t-1}$ is independently sampled,
\[
    \mathrm{Var}(\hat{\mathbf{P}}_{t-1}(e))=\frac{\mathbf{P}(e)(1-\mathbf{P}(e))}{t-1}.
\]
Thus, 
\[
    \mathbb{E}\!\left[\chi^2(\hat{\mathbf{P}}_{t-1},\mathbf{P})\right]
    =
    \sum_{e:\mathbf{P}(e)>0}\frac{\mathrm{Var}(\hat{\mathbf{P}}_{t-1}(e))}{\mathbf{P}(e)}
    =
    \frac{1}{t-1}\sum_{e:\mathbf{P}(e)>0}(1-\mathbf{P}(e))
    \leq
    \frac{m-1}{t-1}.
\]
If $\mathbf{P}(e) = 0$ for some $e\in E$, then the corresponding $\hat{\mathbf{P}}_{t-1}$ also equals 0. Hence, 
\begin{equation*}
    \|\hat{\mathbf{P}}_{t-1}-\mathbf{P}\|_1
=
\sum_{e:\mathbf{P}(e)>0}|\hat{\mathbf{P}}_{t-1}(e)-\mathbf{P}(e)|
\end{equation*}
By the Cauchy--Schwarz inequality, 
\begin{equation*}
    \begin{aligned}
        \|\hat{\mathbf{P}}_{t-1}-\mathbf{P}\|_1
        =
        \sum_{e:\mathbf{P}(e)>0}|\hat{\mathbf{P}}_{t-1}(e)-\mathbf{P}(e)|
        \le
        \sqrt{\sum_{e:\mathbf{P}(e)>0}\mathbf{P}(e)}\;
        \sqrt{\chi^2(\hat{\mathbf{P}}_{t-1},\mathbf{P})}
        =
        \sqrt{\chi^2(\hat{\mathbf{P}}_{t-1},\mathbf{P})}.
    \end{aligned}
\end{equation*}
Taking expectations and applying Jensen's inequality gives
\[
    \mathbb{E}\|\hat{\mathbf{P}}_{t-1}-\mathbf{P}\|_1
    \le
    \sqrt{\mathbb{E}\chi^2(\hat{\mathbf{P}}_{t-1},\mathbf{P})}
    \le
    \sqrt{\frac{m-1}{t-1}}.
\]
\end{proof}

\subsection{Proof of Section~\ref{sec:online_state_feedback}}\label{ap:proof_of_online_state_feedback}

\begin{proof}[Proof of Theorem~\ref{thm_contextual_regret_full_feedback}]
    For each forecast profile $e\in E$, let
    \[
        T_e := \sum_{t=1}^T \mathbf{1}\{e_t=e\}
    \]
    denote the number of days in which context $e$ appears. For rewards in $[-U,U]$, exponential weights with a fixed learning rate $\eta$ satisfies, on the subsequence for context $e$,
    \[
        \max_{a\in\mathcal A}\sum_{t:e_t=e}u(a,\omega_t)
        -
        \mathbb E\left[\sum_{t:e_t=e}u(a_t,\omega_t)\right]
        \le
        \frac{\log d}{\eta}+2\eta U^2T_e.
    \]
    Summing over $e\in E$ gives
    \[
        \mathrm{Reg}^\prime
        \le
        \frac{m\log d}{\eta}+2\eta U^2T.
    \]
    With $\eta=\frac{1}{U}\sqrt{\frac{m\log d}{2T}}$, this becomes
    \[
        \mathrm{Reg}^\prime
        \le
        2U\sqrt{2mT\log d}.
    \]
\end{proof}

\end{document}